\documentclass[lettersize,journal]{IEEEtran}

\usepackage{amsmath,amssymb,amsfonts,amsthm,mathtools,bm}
\usepackage{array}
\usepackage{booktabs}
\usepackage{graphicx}
\usepackage{xcolor}
\usepackage{cite}
\usepackage{stfloats}
\usepackage{url}
\usepackage[caption=false,font=footnotesize]{subfig} 
\usepackage{microtype}
\usepackage{comment}

\usepackage[colorlinks,
linkcolor=blue,       
anchorcolor=blue,     
citecolor=blue,       
]{hyperref}

\newtheorem{theorem}{Theorem}
\newtheorem{proposition}{Proposition}
\newtheorem{lemma}{Lemma}
\newtheorem{corollary}{Corollary}

\theoremstyle{definition}

\newtheorem{definition}{Definition}
\newtheorem{assumption}{Assumption}

\theoremstyle{remark}
\newtheorem{remark}{Remark}

\newcommand{\ms}{\mathsf}

\newcommand{\BE}{\begin{equation}}
	\newcommand{\EE}{\end{equation}}
\newcommand{\BS}{\begin{subequations}}
	\newcommand{\ES}{\end{subequations}}

\newcommand{\dd}{\mathrm{d}}
\newcommand{\ii}{\mathrm{i}}
\newcommand{\RR}{\mathbb{R}}

\newcommand{\bs}[1]{\boldsymbol{#1}}

\newcommand{\half}{p_{\tiny{\ms{half}}}}
\newcommand{\aps}{\rho}

\newcommand{\diag}{\mathrm{diag}}

\begin{document}

\title{Continuous Angular Power Spectrum Recovery From Channel Covariance via Chebyshev Polynomials}

\author{Shengsong Luo,
        Ruilin Wu,
        Chongbin Xu,~\IEEEmembership{Member,~IEEE,}
        Junjie Ma,
        \\
        Xiaojun Yuan,~\IEEEmembership{Fellow,~IEEE,}
        and Xin Wang,~\IEEEmembership{Fellow,~IEEE,}
\thanks{S. Luo, C. Xu, and X. Wang are with the College of Future Information Technology, Fudan University, Shanghai 200433, China (e-mails: \{22110720120\}@m.fudan.edu.cn, \{chbinxu, xwang11\}@fudan.edu.cn).}
\thanks{R. Wu is with the School of Mathematical Sciences, Peking University, Beijing 100871, China (e-mail: rlwu22@stu.pku.edu.cn).}
\thanks{J. Ma is with the State Key Laboratory of Mathematical Sciences, Academy of Mathematics and Systems Science, Chinese Academy of Sciences, Beijing 100190, China (e-mail: majunjie@lsec.cc.ac.cn).}
\thanks{X. Yuan is with the National Key Laboratory of Wireless Communications, University of Electronic Science and Technology of China, Chengdu 611731, China (e-mail: xjyuan@uestc.edu.cn).}
}



\maketitle

\begin{abstract}
This paper proposes a Chebyshev polynomial expansion framework for the
recovery of a continuous angular power spectrum (APS) from channel covariance. By exploiting the orthogonality of
Chebyshev polynomials in a transformed domain, we derive an
exact series representation of the covariance and reformulate
the inherently ill-posed APS inversion as a finite-dimensional linear
regression problem via truncation. The associated approximation error
is directly controlled by the tail of the APS’s Chebyshev series and
decays rapidly with increasing angular smoothness.
Building on this representation, we derive an exact semidefinite
characterization of nonnegative APS and introduce a
derivative-based regularizer that promotes smoothly varying APS profiles while preserving transitions of clusters. 
Simulation results show
that the proposed Chebyshev-based framework yields accurate APS
reconstruction, and enables reliable
downlink (DL) covariance prediction from uplink (UL) measurements in a frequency division duplex (FDD) setting. These findings indicate that jointly exploiting smoothness and nonnegativity in a Chebyshev domain provides an effective tool for covariance-domain processing in multi-antenna systems.
\end{abstract}


\section{Introduction}
\IEEEPARstart{T}{his} paper addresses the problem of recovering the continuous angular power spectrum (APS) from channel covariance.
In multi-antenna/massive-MIMO systems, the APS of the impinging wavefield provides a compact and physically meaningful description of the scattering environment. A reliable estimate of the APS underpins a variety of signal processing and communication tasks, including DL channel covariance conversion in frequency division systems (FDD)~\cite{Miretti-2021,Miretti-2018,Haghighatshoar-2018,Barzegar-2019,Bameri-2023}, channel estimation~\cite{fan2018angle, cui2022channel}, beamforming design~\cite{mohammadzadeh2022covariance,shakya2025angular}, pilot decontamination~\cite{Cavalcante-2020}, power allocation~\cite{kelner2019evaluation}, user grouping~\cite{haghighatshoar2016massive}, and angle-aware localization and sensing~\cite{fischer2025systematic}, among others.

For antenna arrays such as uniform linear array (ULA), the spatial covariance is determined by the APS through a linear integral relation~\cite{sayeed2002deconstructing, haghighatshoar2016massive}, and is statistically stationary over a period of time~\cite{Miretti-2021,Haghighatshoar-2018,Bameri-2023,You-2015}. This has led to a body of work on recovering continuous APS from channel covariance.

A widely used strategy discretizes the angular domain and approximates the continuous integral model by a finite dictionary on a grid, leading to a nonnegative least-squares (NNLS) formulation~\cite{Haghighatshoar-2018, Barzegar-2019,Cavalcante-Renato-2020,song2020deep}. This class of methods is simple to implement, and nonnegativity can be enforced directly on the discretized spectrum. However, the discretization inherently introduces modeling mismatch, and the reconstruction quality can be sensitive to the chosen resolution.

A particularly influential and pioneering line of work, which also provides a main motivation for the present study, is due to
Miretti \emph{et al.}~\cite{Miretti-2021,Miretti-2018}. 
Their projection-onto-a-linear-variety (PLV) framework avoids angular gridding by working directly with continuous APS models. 
PLV treats the APS as an element of $L^2$ space and imposes the covariance constraints by projecting the origin onto an infinite-dimensional affine (linear-variety) set determined by the covariance coefficients. 
Building on this viewpoint, APS nonnegativity is incorporated via projection onto the infinite-dimensional nonnegative cone, yielding an iterative projection scheme, that is the extrapolated alternating projection method (EAPM)~\cite{Bauschke2006Extrapolation}. 
Nevertheless, the EAPM, as an extension of PLV, operates in an \emph{infinite-dimensional} space, still requires fine gridding and burdensome quadrature, and relies on \emph{trigonometric} bases with implicit periodic extension (refer to Proposition~\ref{prop:PLV-equivalent}), which may cause boundary instability when APS components lie near the edges.

More recently, some kernel-based APS estimators have been introduced~\cite{Agrawal-2021,Ninomiya-2022,Kaneko-2024} (often through heuristic Gaussian-mixture models for clustered spectra), which heuristically incorporate the APS smoothness. Yet, its role in modeling error is rarely made explicit in a functional-analytic sense.

We also note that estimating a \emph{continuous} APS from a finite set of covariance lags is an inherently ill-posed inverse problem~\cite{Cavalcante-Renato-2020,Kaneko-2024}: Different angular distributions can generate nearly indistinguishable covariances, especially with a limited number of antennas or closely spaced clusters. To the best of our knowledge, this ill-posedness has not been characterized from an analytical standpoint in the existing literature.
 
 
To address these limitations, we revisit continuous APS recovery from a functional-analytic viewpoint motivated by ideas from functional linear regression~\cite{Interpretable-FLR-2009,morris2015functional}, and develop a framework based on Chebyshev polynomial expansions on a finite interval. By mapping the APS to a \emph{weighted} $L^2$ space on $[-1,1]$ and expanding it in an orthogonal Chebyshev basis~\cite{TrefethenATAP,Plonka2023Chebyshev}, we obtain an exact series representation of the covariance sequence in terms of Chebyshev coefficients. Truncating this series yields a finite-dimensional linear regression model that cleanly separates the array geometry from the APS parametrization and explicitly ties the approximation error to the tail of the Chebyshev series, and hence to the smoothness of the APS (cf.~\cite[Theorem 6.12]{Plonka2023Chebyshev}).
{Moreover, this representation makes transparent how the ill-posedness depends on the number of antennas and the APS regularity.}
Building on this representation, we then incorporate additional structure in a convex and analytically tractable way. First, drawing on classical results on nonnegative polynomials~\cite[Theorem. 1.21.1]{szego2003orthogonal}, we derive an exact semidefinite characterization of nonnegative APS in the Chebyshev domain, leading to a semidefinite programming (SDP) formulation that enforces nonnegativity without angular gridding. 
{Second, inspired by analysis-type sparse regularization~\cite{Elad2007,Stefan2010,bredies2010total,chan2000high}, we design a multi-order finite-difference regularizer in the angular domain that promotes smoothly varying APS profiles while preserving transitions of clusters.}

The main \emph{contributions} of our work are as follows:
\begin{itemize}
	\item \emph{Chebyshev-based representation and its advantages:}
	We introduce a Chebyshev polynomial expansion of a transformed APS on a finite interval, under which the covariance admits a linear model in the Chebyshev coefficients. This representation decouples the array geometry from the APnS parametrization and enables convex covariance fitting. {Moreover, it provides a unified functional-analytic view of PLV/EAPM: While classical PLV/EAPM schemes rely on trigonometric bases with an implicit periodic extension of the APS, our Chebyshev formulation yields (i) explicit truncation-error control via Chebyshev tail bounds for smooth APS, and (ii) avoidance of periodic-extension-induced oscillations.}

    \item 
    \emph{Analytical nonnegativity characterization:}
    We derive an exact semidefinite characterization of nonnegative APS in
    the Chebyshev domain via a structured matrix parametrization of
    the coefficients. This leads to a convex covariance-fitting
    formulation that enforces nonnegativity without angular gridding.
    Arguably, it is more effective than the prior EAPM~\cite{Miretti-2021,Miretti-2018} and NNLS~\cite{Haghighatshoar-2018,Barzegar-2019,Cavalcante-Renato-2020,song2020deep} approaches, which enforce nonnegativity on a (fine) discrete grid.
	
	\item
    \emph{Regularization of analysis-prior:}
    We propose an interpretable regularization framework for APS recovery under limited aperture.
    Our key idea is to encode the clustered nature of APS profiles by combining multi-high-order angular differences into a window-like operator. Penalizing the output of this operator with an $\ell_1$ norm suppresses ripples while retaining a few localized transitions (cf. $\ell_1$ trend-filtering in~\cite{kim2009ell_1}), yielding spectra that are smoothly varying yet preserve transitions of clusters. 

\end{itemize}

{We note that the introduced analysis-prior is shape-aware (sensitive to curvature-like changes), and integrates naturally with physical APS nonnegativity, resulting in a convex and numerically efficient estimator. The numerical results validate the effectiveness of physically plausible APS recovery and its application pertaining to DL channel covariance conversion in a realistic FDD setting.}


\subsection{Preliminary}
Let $I\subseteq\mathbb R$ be a Lebesgue-measurable set and $\mathrm dx$ be the Lebesgue measure. We introduce the following preliminaries.

\begin{definition}
	\label{def:weighted-L2}
	\textit{(Weighted $L^{2}$ space)}.
	Let $w: I\to(0,\infty)$ be a Lebesgue-measurable weight function with $w>0$ almost everywhere.
	For measurable functions $f,g:I\to\mathbb R$, define the weighted inner product
	\(
	\langle f, g \rangle_{w} := \int_{I} f(x)\, g(x)\, w(x)\, \mathrm{d}x,
	\)
	whenever it is well-defined.
	Define the weighted $L^{2}$ space
	\begin{align*}
		L^{2}_{w}( I)
		:= & \bigl\{\, f: I\to\mathbb R \ \big|\ \langle f,f\rangle_{w} < \infty \,\bigr\} ,
	\end{align*}
	with the inner product $\langle \cdot,\cdot\rangle_w$ and norm
	\(
	\|f\|_{w}:=\sqrt{\langle f, f\rangle_{w}}.
	\)
\end{definition}

\begin{definition}
\label{def:projection}
\emph{(Projection Operator).}
Let $\mathcal{S}\subset L^2_w$ be a nonempty closed convex set. The
\emph{projection} of $f\in L^2_w$ onto $\mathcal{S}$ is the (unique)
element $\mathbb{P}_{\mathcal{S}}(f)\in\mathcal{S}$ that solves
\begin{equation}
  \mathbb{P}_{\mathcal{S}}(f)
  := \arg\min_{z\in\mathcal{S}} \,\|f - z\|_{w}^2.
\end{equation}
The mapping $\mathbb{P}_{\mathcal{S}}:L^2_w\to L^2_w$ is called the
\emph{projection operator} onto $\mathcal{S}$. 
In particular, if $\mathcal{S}$ is a closed subspace, then $\mathbb{P}_{\mathcal{S}}$
is the orthogonal projection onto $\mathcal{S}$.
\end{definition}

\begin{lemma}
	\label{def:Chebyshev-Polynomials}
	\textit{(Chebyshev Polynomials~\cite[Sec. 6.1]{Plonka2023Chebyshev}).} The Chebyshev polynomial of degree $n$ is defined by
	\(
	T_n(x) := \cos( n \arccos(x) ),  x \in [-1,1] .
	\)
	The sequence $\{T_n\}_{n\ge0}$ forms a {complete orthogonal basis} for $L^2_w([-1,1])$~\cite[Theorem~6.3]{Plonka2023Chebyshev}, with respect to the inner product $\langle f_1, f_2 \rangle_w$, where the weight function is $w(x) = (1 - x^2)^{-1/2}$.
	That is,
	\begin{equation}
		\label{eq:OrthoPoly}
		\langle T_m, T_n \rangle_{w} =
		\begin{cases}
			\pi & \text{if  } m=n= 0, \\
			\frac{\pi}{2} \delta_{mn} & \text{otherwise}.
		\end{cases}
	\end{equation}
\end{lemma}

\begin{lemma}
	\label{prop:convergence}
	(\textit{Chebyshev Series}\cite[Theorem 3.1]{TrefethenATAP}). If $f$ is Lipschitz continuous on $[-1,1]$~\cite[Sec. 4.14]{Apostol}, it has a unique representation as a Chebyshev series, i.e., 
	\begin{subequations}
	\begin{align}
		\label{eq:c}
		f(x) &= \sum_{n=0}^{\infty} a_n T_n(x), \\
		a_n &= \frac{2-\delta_{n0}}{\pi}\langle f, T_n \rangle_w, 
	\end{align}
	\end{subequations}
	with $f$ being absolutely and uniformly convergent.
\end{lemma}

\begin{definition}
\label{def:zeros-nodes}
	(\textit{Zero Nodes of Chebyshev Polynomials}~\cite{Plonka2023Chebyshev}) 
    The degree-$N$ Chebyshev polynomial $T_N$ defined in Lemma~\ref{def:Chebyshev-Polynomials}, possesses $N$ simple zeros such that
    \(
        T_N(\nu_n) = 0
    \)
    with
    \(
    \nu_n = \cos \frac{2n+1}{2N} \pi,  n=0,1,\ldots,N-1.
    \)
\end{definition}

\emph{Organization}.
The remainder of this paper is organized as follows. Section~\ref{sec:system-model}
introduces the channel covariance model and problem formulation. Section~\ref{sec:cheb-bessel} derives
the Chebyshev-based covariance representation, presents the
truncated linear regression model, and revisits the PLV method
within the same functional setting. Section~\ref{sec:Interpretable} develops the
Chebyshev--SDP framework for nonnegative APS recovery and
incorporates a derivative-based regularizer. Section~\ref{Sec:simulation}
presents simulation results and Section~\ref{sec:conclusion} concludes the
paper.

\emph{Notation}. 
For $\bs{a} \in \mathbb{R}^n$,
$\bs{a}(j_1:j_2)$ denotes the subvector of $\bs{a}$ consisting of its
$j_1$-th to $j_2$-th entries, where $0 \le j_1 \le j_2 \le n-1$; in
particular, $\bs{a}[j]$ denotes the $j$-th entry of $\bs{a}$, and we
also write $a_j$ when no ambiguity can arise. For a matrix
$\bs{A} \in \mathbb{R}^{m \times n}$, $[\bs{A}]_{i,j}$ denotes its
$(i,j)$-th entry, and
$\bs{A}(i_1\!:\!i_2,\,j_1\!:\!j_2)$ denotes the submatrix formed by
rows $i_1$ to $i_2$ and columns $j_1$ to $j_2$, where
$0 \le i_1 \le i_2 \le m-1$ and $0 \le j_1 \le j_2 \le n-1$. 
We use $\mathcal{U}(a,b)$ to denote the uniform distribution on $[a,b]$.


\section{Covariance Model and Problem Description}
\label{sec:system-model}

In this section, we introduce the channel covariance model and formulate the APS recovery problem.

\subsection{Covariance Model}

We consider a single-cell uplink system with a base station (BS) equipped with a ULA of
$M$ antennas and single-antenna user equipment.
Let $\theta\in \Theta:=[-\pi/2,\pi/2]$ be the angle of arrival (AoA).
The ULA steering vector $\bm{a}(\theta) \in \mathbb C^{M \times 1}$ is
\begin{align}
	\label{eq:steering}
	\bs{a}(\theta)
	:=[
	1,
	e^{\ii \kappa_1 \sin\theta},
	\ldots,
	e^{\ii \kappa_{M-1}\sin\theta}
	]^T,
\end{align}
where $\kappa_m := \gamma \pi m$, $m \in \{0,1,\ldots,M-1\}$, with $\gamma=2d/\lambda_{f}$. Here, $d$ is the inter-element spacing of the ULA, and $\lambda_{f}$ is the wavelength of carrier frequency $f$.

We adopt a wide-sense stationary uncorrelated scattering model on angles \cite{You-2015, Haghighatshoar-2018,Bameri-2023}.
{Let $\rho: \Theta \rightarrow \mathbb R_{+}$ denote the APS,
which satisfies the following assumption:}
\begin{assumption}
	\label{assump:APS-continuity}
	\textit{(Regularity of APS).} $\aps$ is Lipschitz continuous. 
\end{assumption}

Following the pioneering works~\cite{Haghighatshoar-2018,Barzegar-2019,Miretti-2021, Bameri-2023}, the spatial covariance matrix can be written as 
\begin{equation}
	\label{eq:R_representation}
	\bs{R}_{}
	=
	\int_{\Theta}
	\rho(\theta)\,
	\bs{a}(\theta)
	\bs{a}(\theta)^H 
	\dd\theta.
\end{equation}
We note that $\bs{R} \in \mathbb C^{M \times M}$ is Hermitian Toeplitz and positive semi-definite. 
This property implies that $\bs{R}$ is determined by their first column $\bm{r}_{}$:
\begin{equation}
	\label{eq:r-vec-exp}
	\bm{r}_{} := [\bs{R}_{}]_{:,1} = \int_{\Theta}
	\rho(\theta)\,
	\bs{a}(\theta) \dd \theta,
\end{equation}
which fully captures the covariance information.

\subsection{Problem Description}

Our goal is to estimate $\rho(\theta)$ based on $\bm{r}$.
Using \eqref{eq:r-vec-exp}, for $\bm{r}=[r_0, r_1, \ldots, r_{M-1}]^\top \in \mathbb C^{M}$, we have,
\begin{equation}
	\label{eq:r_m_exp}
	\begin{aligned}
		r_m &= \int_{\Theta} \aps(\theta) \exp(\ii  \kappa_m \sin \theta) \dd \theta, \quad  m=0,\ldots, M-1.
	\end{aligned}
\end{equation}
Let $x=\sin \theta$ and define $g(x) := \rho(\arcsin x)$, we have
\begin{align}
	\label{eq:integral_g}
	r_m &= \int_{-1}^1 g(x)\,e^{\ii\kappa_m x}\,w(x)\,\dd x
	=\langle g, e^{\ii\kappa_m (\cdot)}\rangle_w,
\end{align}
with the Jacobian term
\(
w(x) = {1}/{\sqrt{1-x^2}}.
\)

We note that the Jacobian $w(x)$ from the change of variables coincides with the weight function of the Chebyshev polynomials.
And Assumption~\ref{assump:APS-continuity} guarantees that $g(x)$ is {Lipschitz continuous} and therefore admits a {Chebyshev polynomial expansion} (cf.~Lemma~\ref{prop:convergence}). 
Since $\rho$ and $g$ are in one-to-one correspondence, estimating $\rho$ is equivalent to estimating $g$; thus, the term ``APS'' is used interchangeably for both.

The mapping $\mathbf{r} \mapsto g$ is inherently ill-posed, as it requires estimating a continuous function from finite samples.
Leveraging Chebyshev modeling provides a structured/robust framework to address this challenge, offering the following advantages:
\begin{itemize}
	\item \emph{Series representation of $r$:} As shown in Section~\ref{sec:cheb-bessel}, properties of Chebyshev expansions and Bessel functions enable a compact series representation of the vector~$\bm{r}$.
	\item \emph{Basis representation of $g$:} The APS~$g$ can be expressed using orthogonal Chebyshev polynomials. Together with the series form of~$r$, this allows the APS recovery to be reformulated as a finite linear regression problem. 
	\item \emph{Oscillation handling:} The oscillatory term $\kappa_m \sin \theta$ in~\eqref{eq:r_m_exp} becomes part of the weight function in~\eqref{eq:integral_g}, yielding a stable Chebyshev representation.
\end{itemize}

\section{Covariance Representation via Truncated Chebyshev Expansion}
\label{sec:cheb-bessel}

In this section, using Chebyshev expansions, we obtain an explicit series form of the covariance and a truncated approximation that leads to a finite-dimensional linear model in the Chebyshev coefficients. 
We further relate this representation to the PLV formulation~\cite{Miretti-2021}.

\subsection{Covariance Representation and Truncation}
\label{subsec:cheb-trunc}


Consider the Chebyshev polynomial expansion of a Lipschitz continuous APS $g \in L^2_w$ (cf.~\eqref{eq:integral_g}):
\begin{align}
	\label{eq:cheb-expansion}
	g(x)=\sum_{n\geq0}^{} a_n\,T_n(x),
	\quad
\end{align}
in which $T_n(x) = \cos( n \arccos(x) )$,
\(
	a_n=\frac{2-\delta_{n0}}{\pi}\!\langle g, T_n \rangle_{w}
	\label{eq:cheb-coeff}
\)
(cf. Lemma~\ref{def:Chebyshev-Polynomials}).
With slight abuse of notation, we define $r(\kappa)$ as
\begin{equation}
	\label{eq:r-fun}
	r(\kappa) := \langle g, e^{\ii\kappa (\cdot)}\rangle_w, \quad I:=[-1,1].
\end{equation}
\begin{remark}
\label{remark:spatial-frequency-samples}
	The sequence $\{r_m\}_{m=0}^{M-1}$
	in \eqref{eq:r_m_exp} can be viewed as discrete samples of this function at points $\{\kappa_m\}_{m=0}^{M-1}$.
	In Theorem~\ref{theorem:CB-parseval}, it is shown that the function $r(\kappa)$ admits a series expansion using Chebyshev polynomials.
\end{remark}


\begin{theorem}
	\label{theorem:CB-parseval}
	If Assumption~\ref{assump:APS-continuity} holds, then $g \in L^2_w$ is Lipschitz continuous with Chebyshev expansion $g(x)=\sum_{n=0}^\infty a_n T_n(x)$, where $\{a_n\}$ is the Chebyshev coefficients of $g$. Recall the definition of $r(\kappa)$ in \eqref{eq:r-fun}.
	Then, the following holds true.
	
	\textbf{(1) Series representation of $r$.} Let $\ii = \sqrt{-1}$,
	\begin{equation}
		\label{eq:cheby-Bessel}
		r(\kappa)
		= \pi \sum_{n=0}^\infty \mathrm i^n a_n J_n(\kappa),
	\end{equation}
	where $J_n(\cdot)$ is the Bessel function of the first kind of order $n$. 
	
	\textbf{(2) Uniform bound on truncation error.}
	For $p \in\mathbb N_{>1}$,
	let $g_{\le p}(x):=\sum_{n=0}^p a_n T_n(x)$ be the $p$-th order Chebyshev sum. We denote
	\begin{equation}
		\label{eq:r-P-exp}
		r_{\le p}(\kappa) := \int_{I} g_{\le p}(x) e^{\ii\kappa x}\,w(x)\,\dd x 
		=\pi\sum_{n=0}^p \mathrm i^n a_n J_n(\kappa).
	\end{equation}
	The truncation error is uniformly (in $\kappa$) bounded by
	\begin{align}
		|r(\kappa)-r_{\le p}(\kappa) |
		\;\le\;
		\sqrt{\pi}\;\|\,g-g_{\le p}\,\|_{w},
		\quad \forall \kappa \in \RR,
		\label{eq:tail-function}
	\end{align}
	with 
    \(
		\|\,g-g_{\le p}\,\|_{w} = 
		\sqrt{{\pi}/{2}}
		\big(\sum_{n>p} a_n^2\big)^{\!1/2}.
    \)
\end{theorem}
\begin{proof}
	See Appendix~\ref{appendix-pf-theorem:CB-parseval}.
\end{proof}


Fig.~\ref{relation} shows the relationship among various functions/series.
Some \emph{comments} concerned with Theorem~\ref{theorem:CB-parseval} are in order.


First, Theorem \ref{theorem:CB-parseval}-(1) establishes an explicit series representation of the covariance directly in terms of the Chebyshev coefficients of the APS $g$. This derivation naturally inspires a truncation strategy to reformulate the covariance fitting problem as a finite linear system (cf. \cite{Interpretable-FLR-2009, morris2015functional}), enabling the subsequent linear regression formulation in Section~\ref{sec:APS_acquisition}.

Second, the covariance fitting error in the Chebyshev framework is fundamentally controlled by the Chebyshev tail coefficients $\{a_n\}_{p+1}^{\infty}$, which directly dictates how fast the error vanishes with the truncation order $p$. This relationship is formalized by the Chebyshev tail (cf. \eqref{eq:tail-function}). 
In other words, if the $p$-th order Chebyshev $g_{\le p}$ approximates $g$ well in the Chebyshev weighted norm, then the corresponding $r_{\le p}$ uniformly approximates $r$ with only a loss factor $\sqrt{\pi}$.

\begin{remark}
	\label{remark:smoothness}
Note that classical results on Chebyshev approximation show that the
truncation error in~\eqref{eq:tail-function} is entirely governed by the \emph{smoothness} of the APS $g$~\cite{Plonka2023Chebyshev,TrefethenATAP}. Specifically, if $g$ has $s$ bounded derivatives on $I$, then the covariance fitting error
$|r(\kappa) - r_{\le p}(\kappa)|$ decreases at least at a
polynomial rate as the truncation order $p$ increases; if $g$
is smooth enough, namely analytic (cf.~\cite[Ch. 8.1]{TrefethenATAP}), the error decays
exponentially fast in $p$ (see, e.g.,~\cite[Theorem 6.12]{Plonka2023Chebyshev} and~\cite[Theorem 8.1]{TrefethenATAP} for precise statements). Hence, a
smoother APS leads to much faster convergence and allows
accurate covariance modeling with a low truncation order $p$, which highlights that our fitting framework explicitly exploits the \emph{smoothness} of the APS. 
\end{remark}

The regularity of the APS also admits a physical interpretation:
A smoother APS typically corresponds to richer
scattering environments, which in turn lead to more stable and accurate
APS reconstruction. The smoothness assumptions on the APS have also
been discussed in~\cite{Agrawal-2021,Ninomiya-2022,Kaneko-2024}.

The above discussion motivates a parametric representation of the
channel covariance by a finite number of Chebyshev
coefficients of the APS. With a suitable truncation order $p$, the
relationship between these coefficients and the covariance
lags can be written in a linear form. We next exploit this structure
to cast covariance fitting as a linear regression problem.

\begin{figure}[t]
\centering
\includegraphics[width=0.9\linewidth]{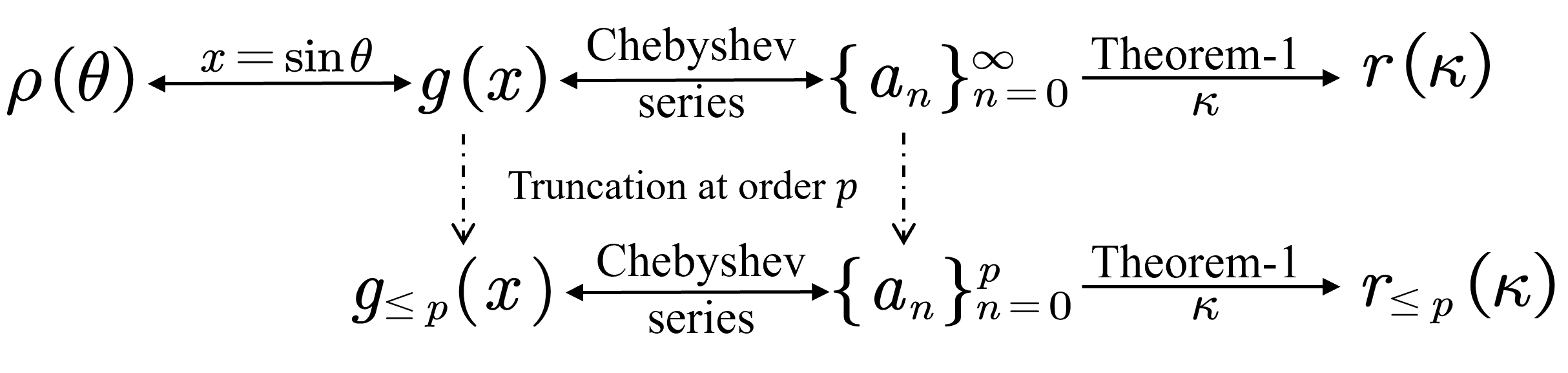}
\caption{Relations among channel covariance, (truncated) APS functions, and the corresponding Chebyshev series.}
\label{relation}
\end{figure}

\subsection{Linear Regression Formulation} 
\label{sec:APS_acquisition}

Recall the intermediate function $g(x) = \rho(\arcsin x)$ with
$x = \sin \theta$. The Chebyshev coefficients $\{a_n\}$ of $g$
provide an equivalent parametric representation of the APS
$\rho(\theta)$:
\begin{equation}
	\label{eq:rho-expansion}
	\rho(\theta) = g(\sin\theta) = \sum_{n\geq0} a_n T_n(\sin\theta).
\end{equation}
Hence, estimating the APS $\rho(\theta)$ is equivalent to
estimating the coefficient sequence $\{a_n\}$. Combining this
representation with the Chebyshev-series expression of the
covariance sequence in~\eqref{eq:cheby-Bessel} yields a linear relationship between
$\{a_n\}$ and the covariance lags $\{r_m\}$.

Expanding $g(x) = \sum_{n\ge 0} a_n T_n(x)$ in $L^2_w$ and
using Theorem~\ref{theorem:CB-parseval} give the exact series representation
\begin{subequations} 
\label{eq:r-Bessel-expansions}
\begin{align}
  \Re(r_m)
  &= \pi \sum_{n\ge 0} (-1)^n a_{2n} J_{2n}(\kappa_m),
  \label{eq:real-part}\\
  \Im(r_m)
  &= \pi \sum_{n\ge 0} (-1)^n a_{2n+1} J_{2n+1}(\kappa_m),
  \label{eq:imag-part}
\end{align}
\end{subequations}
for $m = 0,\ldots,M-1$. Motivated by
the smoothness properties of the APS discussed above, we approximate the series in~\eqref{eq:r-Bessel-expansions} by truncating
at an odd order $p>1$. The first $p+1$ coefficients
$\{a_n\}_{n=0}^{p}$ are retained as unknown parameters, whereas
the remaining tail is absorbed into a modeling error term.

Let $p_{\mathrm{half}} := \lfloor p/2 \rfloor$. We collect the
\textbf{e}ven- and \textbf{o}dd-indexed coefficients into
\begin{equation}
\label{eq:aeven-aodd-def}
\begin{aligned}
  \bs{a}_{\mathrm{e}}
  &:= \big[\sqrt{2}\,a_0,\, a_2,\,\ldots,\,a_{p-1}\big]^\top
      \in \mathbb{R}^{p_{\mathrm{half}}}, \\
  \bs{a}_{\mathrm{o}}
  &:= \big[a_1,\,a_3,\,\ldots,\,a_p\big]^\top
      \in \mathbb{R}^{p_{\mathrm{half}}}.
\end{aligned}
\end{equation}
Stacking the relations in~\eqref{eq:r-Bessel-expansions} for
$m = 0,\ldots,M-1$ then leads to two real-valued linear
regressions,
\begin{subequations}\label{eq:two-regressions}
\begin{align}
  \bs{y}_{\mathrm{real}}
  &= \bs{\Phi}_{\mathrm{e}}\, \bs{a}_{\mathrm{e}}
     + \bs{\epsilon}_{\mathrm{e}} \in \mathbb{R}^{M},
     \label{eq:reg-even}\\
  \bs{y}_{\mathrm{imag}}
  &= \bs{\Phi}_{\mathrm{o}}\, \bs{a}_{\mathrm{o}}
     + \bs{\epsilon}_{\mathrm{o}} \in \mathbb{R}^{M-1},
     \label{eq:reg-odd}
\end{align}
\end{subequations}
where the truncation residuals are given by
\[
\begin{split}
  [\bs{\epsilon}_{\mathrm{e}}]_j
  &= \sum_{n\ge p_{\mathrm{half}}}
       (-1)^n a_{2n} J_{2n}(\kappa_j),
     \quad \qquad  0 \le j \le M-1, \\
  [\bs{\epsilon}_{\mathrm{o}}]_j
  &= \sum_{n\ge p_{\mathrm{half}}}
       (-1)^n a_{2n+1} J_{2n+1}(\kappa_{j+1}),  0\leq  j \le M-2,
\end{split}
\]
and the ``effective observations" are defined as
\begin{subequations}\label{eq:y-def}
\begin{align}
  \bs{y}_{\mathrm{real}}
  &:= \frac{1}{\pi}
      \big[\Re(r_0),\Re(r_1),\ldots,\Re(r_{M-1})\big]^\top,
      \\
  \bs{y}_{\mathrm{imag}}
  &:= \frac{1}{\pi}
      \big[\Im(r_1),\Im(r_2),\ldots,\Im(r_{M-1})\big]^\top.
      \label{eq:y-imag}
\end{align}
\end{subequations}
The design matrices
$\mathbf{\Phi}_{\mathrm{e}}\in\mathbb{R}^{M\times p_{\mathrm{half}}}$,
$\mathbf{\Phi}_{\mathrm{o}}\in\mathbb{R}^{(M-1)\times p_{\mathrm{half}}}$
are defined entry-wise, for $n=0,1,\ldots,p_{\mathrm{half}}-1$, as
\begin{subequations}\label{eq:Phi-def}
\begin{align}
  [\bs{\Phi}_{\mathrm{e}}]_{j,n}
  &:= (-1)^n \frac{1}{\sqrt{1+\delta_{n0}}}J_{2n}(\kappa_j),
      0 \le j \le M-1, \\
  [\bs{\Phi}_{\mathrm{o}}]_{j,n}
  &:= (-1)^n J_{2n+1}(\kappa_{j+1}),
     \quad  0 \le j \le M-2. \label{eq:Phi-o}
\end{align}
\end{subequations}

Note that $\Im(r_0) = 0$, which is consistent with the fact that
$J_{2k+1}(\kappa_0) = 0$ for all $k\in\mathbb{N}$ when $\kappa_0 = 0$~\cite{WatsonBessel}. Hence $\bs{y}_{\mathrm{imag}}$ in~\eqref{eq:y-imag} has length $M-1$ and
$\bs{\Phi}_{\mathrm{o}}$ in~\eqref{eq:Phi-o} has $M-1$ rows. To better
expose the degrees of freedom and obtain a unified real-valued model,
we define:
\begin{align}
\bs{\Phi}&:=\left[
\begin{matrix}
	\bs{\Phi}_{\mathrm{e}} &\bm{0} \\
	\bm{0} & \bs{\Phi}_{\mathrm{o}}
\end{matrix}
\right]
\in \mathbb R^{(2M-1) \times (p+1)} ,
\\
\bs{a}
&:=
[\sqrt{2}a_0, a_1,\ldots,a_p]^\top
\in \mathbb{R}^{p+1},
\label{eq:a-def}
\end{align}
and a fixed permutation matrix $\bs{\Pi}$ such that
\begin{equation}
	\mathbf{\Pi}\,\mathbf{a}
	=
	\begin{bmatrix}
		\bs{a}_{\mathrm{e}} \\
		\bs{a}_{\mathrm{o}}
	\end{bmatrix},
	\label{eq:Pi-def}
\end{equation}
with $\bs{a}_{\mathrm{e}}$ and $\bs{a}_{\mathrm{o}}$ defined in~\eqref{eq:aeven-aodd-def}.

Then, we can stack the two real systems in~\eqref{eq:two-regressions} as 
\begin{align}
	\label{eq:linear-reg}
	\bs y :=
	\begin{bmatrix} \bs{y}_{\mathrm{real}} \\ \bs{y}_{\mathrm{imag}} \end{bmatrix}
	= \bm{\Phi}\, \bs{\Pi} \bs a_{} + \bs \epsilon,
\end{align}
where $\bs{\epsilon}=[\bs{\epsilon}_{\mathrm{e}}; \bs{\epsilon}_{\mathrm{o}}] \in \RR^{2M-1}$ collects the truncation residuals.
Equation~\eqref{eq:linear-reg} is a compact real-valued linear
regression model that links the truncated Chebyshev coefficients of
the APS to the covariance lags.


\begin{corollary}
\label{corollary:Chebyshev}
The APS is perfectly recoverable via the Chebyshev
framework if the following holds:
(1) the APS can be represented by a degree-$p$
Chebyshev polynomial expansion, and (2) the matrix
$\mathbf{\Phi}\in\mathbb{R}^{(2M-1)\times(p+1)}$ has full column
rank. 
\end{corollary}

Based on~\eqref{eq:linear-reg}, we will further explore the implicit structural information of $\bs{a}$ in Section~\ref{sec:Interpretable} to obtain a robust estimation.


\subsection{Connection to PLV Approach}

We revisit the PLV
scheme of~\cite{Miretti-2021,Miretti-2018} and place it in the same transformed-domain
representation $g(x) = \rho(\arcsin x) \in L^2_w(I)$ as in our Chebyshev
framework. This leads to a geometric interpretation of PLV 
and shows that PLV reconstructs the APS within a fixed-order
trigonometric polynomial model. 


As its name implies, the PLV approach projects the origin onto a specific linear variety (i.e., affine set), as formally defined in Lemma~\ref{lem:Affine-sub-space}.
To facilitate analysis and build the connection, we state the following Lemma.

\begin{lemma}[Affine structure and orthogonal decomposition]
\label{lem:Affine-sub-space}
Let $\bs{r} = [r_0,\ldots,r_{M-1}]^{\top} \in \mathbb{C}^{M}$ be a
given covariance vector, and define $r_{-m} := r_m^{*}$ for
$m=1,\ldots,M-1$. Consider the $L^2_w$ space in Definition~\ref{def:weighted-L2}
and define the affine set
\begin{align}
  \mathcal{V}_{\bs{r}}
  &:= \big\{ f  \;|\;
        \langle f, e^{\ii \kappa_m (\cdot)} \rangle_w = r_m,
        \; -M < m < M \big\}, \label{eq:Vr-def}
\end{align}
together with the closed subspace (i.e., null space)
\begin{align}
  \mathcal{N}
  &:= \big\{ f  \;|\;
        \langle f, e^{\ii \kappa_m (\cdot)} \rangle_w = 0,
        \; -M < m < M \big\}, \label{eq:N-def}
\end{align}
and its orthogonal complement
\begin{align}
  \mathcal{N}_{\bot}
  &:= \big\{ f  \;|\;
        \langle f, h \rangle_w = 0,\;\forall\, h \in \mathcal{N} \big\}.
        \label{eq:N-perp-def}
\end{align}
Then the following two statements hold true:
\begin{enumerate}
  \item $\mathcal{V}_{\bs{r}}$ is an affine subset of $L^2_w(I)$
        with direction $\mathcal{N}$, i.e., for any
        $f_1 \in \mathcal{V}_{\bs{r}}$ and
        $f_2 \in \mathcal{N}$ we have $f_1 + f_2 \in \mathcal{V}_{\bs{r}}$,
        and hence
        \(
          \mathcal{V}_{\bs{r}} + \mathcal{N} = \mathcal{V}_{\bs{r}}.
        \)
  \item The orthogonal complement $\mathcal{N}_{\bot}$ is a
        finite-dimensional subspace spanned by trigonometric
        polynomials, i.e.,
        \begin{equation}
          \label{eq:complement-space}
          \mathcal{N}_{\bot}
          = \big\{ f(\cdot;\bs{b}) \;|\;
                     \bs{b} \in \mathbb{R}^{2M-1} \big\},
        \end{equation}
        where
        \[
          f(x;\bs{b})
          = b_0
            + \sum_{m=1}^{M-1}
               b_m \cos(\kappa_m x)
                  + b_{M-1+m} \sin(\kappa_m x) ,
        \]
        with $\kappa_m$ denoting the spatial-frequency samples
        associated with the covariance lags (cf. Remark~\ref{remark:spatial-frequency-samples}).
\end{enumerate}
\end{lemma}
\begin{proof}
    See Appendix~\ref{proof:affine}.
\end{proof}

Lemma~\ref{lem:Affine-sub-space} shows that
$g \in L^2_w(I)$ consistent with the covariance constraints form an
affine set $\mathcal{V}_{\bs r}$ with direction $\mathcal{N}$, and
that $\mathcal{N}_{\bot}$ is a finite-dimensional subspace spanned by
trigonometric polynomials. The PLV~\cite{Miretti-2021} can
thus be viewed as selecting the unique element in
$\mathcal{V}_{\bs r} \cap \mathcal{N}_{\bot}$ (see Fig.~\ref{fig:PLV}). Building
on this geometric picture, the following proposition summarizes
several equivalent characterizations of the PLV solution.

\begin{figure}
	\centering
	\includegraphics[width=0.4\linewidth]{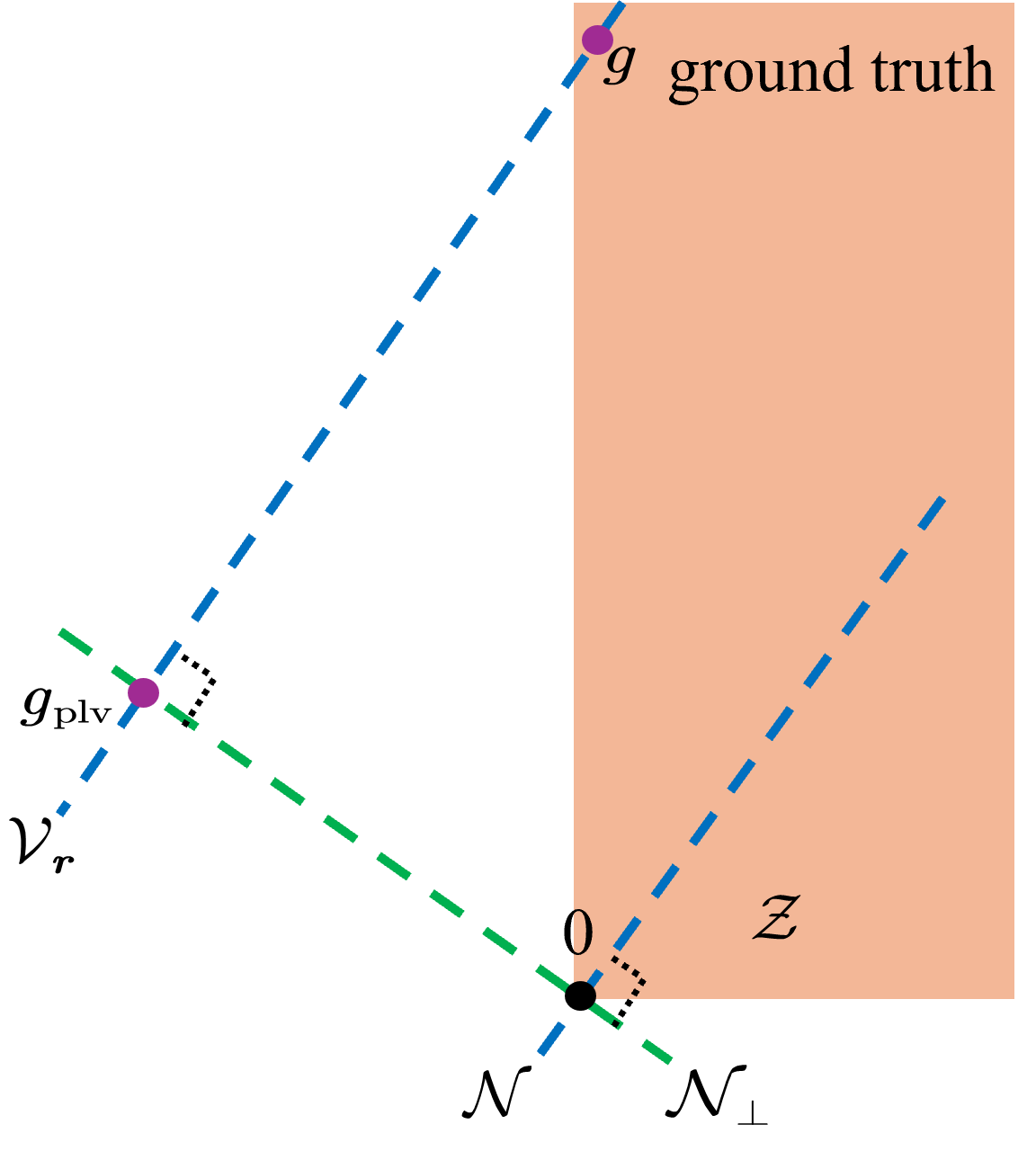}
	\caption{Geometric interpretation of the PLV. The affine set $\mathcal{V}_{\bs r}$ collects all functions consistent with the covariance constraints, $\mathcal{N}$ denotes the associated null space, and $\mathcal{N}_{\bot}$ is its orthogonal
		complement. $\mathcal{Z}$ is defined in Section~\ref{subsec:eapm}.}
	\label{fig:PLV}
\end{figure}

\begin{proposition}[Equivalent characterizations of the PLV]
\label{prop:PLV-equivalent}
Recall the sets $\mathcal{V}_{\bs{r}}$, $\mathcal{N}_{\bot}$ defined in
Lemma~\ref{lem:Affine-sub-space}. Let $g_{\mathrm{plv}} \in L^2_w(I)$ denote the solution of the PLV~\cite{Miretti-2021}. Then the following four
statements are equivalent and uniquely characterize $g_{\mathrm{plv}}$.
\begin{enumerate}
  \item \textbf{Minimum-norm formulation:}
        \begin{equation}
          g_{\mathrm{plv}}
          \;=\;
          \arg\min_{\mathring{g} \in \mathcal{V}_{\bs{r}}}
                 \|\mathring{g}\|_{w}
          \;=\;
          \mathbb{P}_{\mathcal{V}_{\bs{r}}}(0).
          \label{eq:plv-min-norm}
        \end{equation}

  \item \textbf{Geometric feasibility formulation:}
        \begin{equation}
          g_{\mathrm{plv}}=\mathcal{V}_{\bs{r}} \cap \mathcal{N}_{\bot}
          \label{eq:plv-intersection}
        \end{equation}
        i.e., $g_{\mathrm{plv}}$ is the unique element in the
        intersection set.

  \item \textbf{Trigonometric polynomial parameterization:}  
        There is a unique vector
        $\bs{b} = [b_0,\ldots,b_{2M-2}]^{\top} \in \mathbb{R}^{2M-1}$
        such that
        \begin{equation}
          g_{\mathrm{plv}}(x;\bs{b})
          := 
          b_0
          + \sum_{m=1}^{M-1}
             b_m \cos(\kappa_m x)
                + b_{M-1+m} \sin(\kappa_m x) ,
          \label{eq:plv-trig-form}
        \end{equation}
        and $\bs{b}$ solves the following feasibility problem:
        Letting $r_{-m} := r_m^{*}$ for $m=-(M-1),\ldots,M-1$,
        \begin{equation}
          \begin{aligned}
          \text{find } \bs{b} \in \mathbb{R}^{2M-1}
           \text{s.t.}\quad
          \langle g_{\mathrm{plv}}(\cdot;\bs{b}),
                  e^{\ii \kappa_m (\cdot)} \rangle_w
          = r_m. 
          \end{aligned}
          \label{eq:plv-trig-feas}
        \end{equation}

  \item \textbf{Closed-form coefficient representation:}  
        The vector $\bs{b}$ in~\eqref{eq:plv-trig-form} admits the
        closed form
        \begin{equation}
          \bs{b}
          = \pi \mathbf{G}^{-1} \bs{y},
          \label{eq:plv-closed-form}
        \end{equation}
        where $\bs{y}$ is defined
        in~\eqref{eq:linear-reg}, and
        $\mathbf{G}$ satisfies\footnote{Note that matrix $\mathbf G$ slightly differs from the one in~\cite[Proposition~1]{Miretti-2021}; in particular, the latter is rank-deficient, whereas $\mathbf G$ is full-rank in our setting.
        }
        \begin{equation}
          \mathbf{G}
          :=
          \begin{bmatrix}
            \mathbf{G}_{\Re} & \mathbf{0} \\
            \mathbf{0}       & \mathbf{G}_{\Im}
          \end{bmatrix} \in \mathbb{R}^{(2M-1)\times(2M-1)},
          \label{eq:plv-G-def}
        \end{equation}
        with $\mathbf{G}_{\Re} \in \mathbb{R}^{M\times M}$ and
        $\mathbf{G}_{\Im} \in \mathbb{R}^{(M-1)\times(M-1)}$
        given by
       \[
    \begin{split}
    [\mathbf{G}_{\Re}]_{m,n}
      &= \tfrac{\pi}{2}\big(J_0(\kappa_{m-n}) + J_0(\kappa_{m+n})\big), \\
    [\mathbf{G}_{\Im}]_{m^{\prime}-1,n^{\prime}-1}
      &= \tfrac{\pi}{2}\big(J_0(\kappa_{m^{\prime}-n^{\prime}}) - J_0(\kappa_{m^{\prime}+n^{\prime}})\big),
    \end{split}
    \]
    for $0 \le m,n \le M-1$ and $1 \le m^{\prime},n^{\prime} \le M-1$, respectively.
    In particular, $\mathbf{G} \succ \bs{0}$ and hence invertible.
\end{enumerate}
\end{proposition}

\begin{proof}
    See Appendix~\ref{proof:proof-PLV}.
\end{proof}


\begin{corollary}
\label{corollayr:PLV}
	The APS becomes {perfectly recoverable} via the PLV approach if and only if the ground-truth APS $g \in \mathcal{N}_{\bot}$, i.e., is represented by the trigonometric polynomials in~\eqref{eq:plv-trig-form} with order up to $2M-2$ (a total of $2M-1$ degrees).
\end{corollary}

Note that the result of Corollary~\ref{corollayr:PLV} is consistent with Corollary~\ref{corollary:Chebyshev}, in terms of degrees of freedom (model order).

\begin{remark}[Additional insight into the PLV formulation]
Building upon the original PLV method in~\cite{Miretti-2021}, Proposition~\ref{prop:PLV-equivalent} makes
two aspects of the method more explicit.
\begin{itemize}
    \item 
\emph{Appearance of trigonometric polynomials}.
The explicit description of $\mathcal{N}_{\bot}$ shows that
PLV reconstructs the APS within a fixed-order trigonometric polynomial
model in the transformed domain $g \in L^2_w(I)$. In particular, PLV selects the basis
$\{1,\cos(\kappa_m x),\sin(\kappa_m x)\}$ for representing $g$.
\item
\emph{Uniqueness}.
The closed-form relation~\eqref{eq:plv-closed-form} reveals that the PLV
estimate is linear to the covariance, with a
matrix $\mathbf{G}$ that depends only on the array geometry and frequencies. Since $\mathbf{G} \succ\bs{0}$, the coefficients $\bs{b}$, and hence
$g_{\mathrm{plv}}$, are uniquely determined by the covariance
constraints. This makes the uniqueness of the PLV solution explicit.
\end{itemize}
\end{remark}

\textit{Discussion}.
Proposition~\ref{prop:PLV-equivalent} shows that PLV recovers the APS by projecting
onto a fixed dimensional trigonometric polynomial subspace
in the transformed domain $g \in L^2_w(I)$, yielding a unique solution
that depends linearly on the covariance vector. Our Chebyshev-based
framework enforces the similar covariance-fitting constraints but differs
in that (i) it uses a Chebyshev basis on a finite interval instead of a
periodic trigonometric basis, avoiding implicit periodic extension and
wrap-around artifacts; (ii) it ties the APS smoothness directly to the
decay of the Chebyshev coefficients, which naturally guides the
truncation order (see discussion in Remark~\ref{remark:smoothness}). 

\section{Nonnegative and Regularized APS Recovery}
\label{sec:Interpretable}

In this section, we develop a robust, convex Chebyshev-based framework for APS recovery in limited-aperture regimes, where the regression matrix can be underdetermined. We show how APS nonnegativity is enforced analytically within a finite-dimensional parametrization, relate the resulting formulation to projection-based approaches, and then incorporate a derivative-based regularizer to promote APS profiles that are smoothly varying yet exhibit a few localized transition angles.

\subsection{Ill-posedness of APS Recovery}

Recall the linear regression model in~\eqref{eq:linear-reg}. In the limited-aperture regime, accurately representing a continuous APS typically requires a high-order Chebyshev expansion (large $p$ such that $\bs{\epsilon} \rightarrow 0$), often with $p>2M$. This places the problem in an \emph{underdetermined} setting: The designed operator cannot uniquely determine coefficients $\bs{a}$ from the available data.

Equivalently,~\eqref{eq:linear-reg} is not identifiable because the sensing matrix $\bs\Phi$ has a nontrivial null space when $p>2M$. Therefore, there exist infinitely many \emph{data-consistent} coefficient vectors,
\[
\hat{\bs a}=\bs a+\bs u,\quad \bs u\in\mathrm{Null}(\bs\Phi),
\]
all producing the same observation $\bs y$, where $\text{Null}(\bs \Phi)$ is the null-space of $\bs{\Phi}$. In practice, different data-consistent solutions may mainly differ by non-physical oscillations (e.g., angular ripples) that are weakly constrained by the measurements.

To obtain a physically meaningful APS estimate, we incorporate additional structure in the sequel to suppress the null-space, namely nonnegativity and an analysis-type prior that selects a plausible solution among the data-consistent candidates.

\subsection{Semidefinite Characterization of Nonnegative APS}
\label{subsec:SDP}

The APS must be nonnegative on the angular interval in order to admit a
physical plausibility. This requirement translates into a pointwise inequality on
the finite-order approximation, i.e.,
\(
	g_{\leq p} (x; \bs{a}) \geq 0,  x \in I,
\)
which is inconvenient to enforce directly because it involves infinitely many constraints over $x \in I$. To obtain a \emph{tractable} formulation, we exploit classical results on nonnegative polynomials (cf. Proposition~\ref{prop:Lukas}), expressing this constraint through
a pair of positive semidefinite matrices that satisfy linear relations
with the Chebyshev coefficients $\bs{a}$. 

To facilitate analysis, we introduce the following definition.

\begin{definition}
	\label{def:nonnegative-cone}
	\textit{(Chebyshev Nonnegative Cone)}
	Assume $p$ is odd and set $p_{\mathrm{half}} := (p+1)/2$.
	Let $N \ge p$ be an integer and
	$\{\nu_j\}_{j=0}^{N} \subset [-1,1]$ denote the Chebyshev zero nodes
	of the polynomial $T_{N+1}$ (cf. Definition~\ref{def:zeros-nodes}).
	Let $\bs{\Psi} \in \mathbb{R}^{(N+1)\times(N+1)}$ be the orthonormal
	discrete cosine transform of type-II (DCT-II) on the Chebyshev zero
	nodes~\cite[Lemma~3.47]{Plonka2023Chebyshev} (see also~\cite[Sec.~8.5]{briggs1995dft}),
	with entries
	\[
	[\bs\Psi]_{j,n}
	= \sqrt{\frac{2-\delta_{n0}}{N+1}}\,
	T_n(\nu_j),
	\quad 0 \le j,n \le N.
	\]
	We define the Chebyshev nonnegative cone $\mathcal{D} \subset \mathbb{R}^{p+1}$ by\footnote{The nonnegativity interpretation of $\mathcal{D}$ will be made explicit in Lemma~\ref{lem:sdp_cheb}.}
	\begin{equation}
		\mathcal{D}
		:= \big\{ \beta(\bm{S}_1, \bm{S}_2) \in \mathbb{R}^{p+1}
		: \bm{S}_1, \bm{S}_2 \in \mathbb{S}^{p_{\mathrm{half}}}_{+} \big\},
		\label{eq:D-def}
	\end{equation}
	where $\mathbb{S}^{p_{\mathrm{half}}}_+$ is the cone of
	$p_{\mathrm{half}} \times p_{\mathrm{half}}$
	positive semidefinite matrices, and
	$\beta:\mathbb{S}^{p_{\mathrm{half}}}_+ \times
	\mathbb{S}^{p_{\mathrm{half}}}_+ \to \mathbb{R}^{p+1}$ is the linear map
	\begin{equation}
		\label{eq:map-define}
		\begin{aligned}
			\beta(\bm{S}_1, \bm{S}_2)
			:=\;& \bs{\Psi}_{p+1}^\top \left[( \bs{1} + \bm{\nu} )\circ
			\diag(\bs{\Psi}_{\half} \bm{S}_1 \bs{\Psi}_{\half}^\top ) \right] \\
			+&\;  \bs{\Psi}_{p+1}^\top \left[( \bs{1} - \bm{\nu} ) \circ
			\diag(\bs{\Psi}_{\half} \bm{S}_2 \bs{\Psi}_{\half}^\top) \right],
		\end{aligned}
	\end{equation}
	where
	\(
	\bs \Psi_{p_{\mathrm{half}}} := \bs\Psi(:,0:p_{\mathrm{half}}-1),
	\bs \Psi_{p+1} := \bs \Psi(:,0:p),
	\)
	and $\bs{\nu} := (\nu_0,\ldots,\nu_N)^\top$,
    $\circ$ denotes the Hadamard product and
	$\diag(\cdot)$ extracts the diagonal of a square matrix as a vector. 
\end{definition}

Based on Definition~\ref{def:nonnegative-cone},
The following Lemma provides a semidefinite characterization of nonnegative APS.
\begin{lemma}[Semidefinite characterization of nonnegative APS]
	\label{lem:sdp_cheb}
	Assume $p$ is odd.
	Consider the truncated Chebyshev expansion
	\[
	g_{\le p}(x;\bs{a})
	= \sum_{n=0}^{p} a_n T_n(x), 
	\quad x \in I := [-1,1],
	\]
	with coefficient vector $\bs{a} := [\sqrt{2}a_0,\ldots,a_p]^\top \in \mathbb{R}^{p+1}$.
	Recall the Chebyshev nonnegative cone $\mathcal{D}$ defined in
	Definition~\ref{def:nonnegative-cone}. Then the following two
	statements are equivalent:
	\begin{enumerate}
		\item $g_{\le p}(x;\bs{a}) \ge 0$ for all $x \in I$;
		\item $\bs{a} \in \mathcal{D}$, i.e., there exist two positive semidefinite
		matrices $\bm{S}_1,\bm{S}_2\succeq \bs{0}$
		such that $\bs{a} = \beta(\bm{S}_1,\bm{S}_2)$
		(cf.~\eqref{eq:D-def} and~\eqref{eq:map-define}).
	\end{enumerate}
	In particular, the infinite family of pointwise nonnegativity
	constraints $g_{\le p}(x;\bs{a}) \ge 0$ for all $x \in I$ admits an
	exact semidefinite representation in the coefficient space: It is
	equivalent to the existence of decision matrices
	$\bm{S}_1,\bm{S}_2 \succeq \bm{0}$ and the linear relations encoded
	by the map $\beta$.
\end{lemma}
\begin{proof}
	See Appendix~\ref{appendix-pf-lemma-SDP} (the construction in Definition~\ref{def:nonnegative-cone} is in fact guided by the proof of this Lemma).
\end{proof}

\begin{remark}
	Note that the matrix-vector product involving $\bs{\Psi}$ can be evaluated in
	$\mathcal{O}(N \log N)$ operations using FFT-based DCT routines~\cite[Algorithm 6.22]{Plonka2023Chebyshev}, and
	may be convenient for implementing the nonnegativity in the resulting SDP.
\end{remark}

By Lemma~\ref{lem:sdp_cheb}, the nonnegativity constraint $g_{\le p}(x;\bs{a}) \ge 0$ on $I$ is
equivalently enforced by requiring that the Chebyshev coefficient
vector $\bs a$ belongs to the semidefinite-representable set
$\mathcal{D}$ in~\eqref{eq:D-def}. In other words, nonnegativity of
the APS is captured through the existence of two positive semidefinite
matrices $\bs S_1, \bs S_2 \in \mathbb{S}^{p_{\mathrm{half}}}_+$ satisfying
$\bs a = \beta( \bs S_1, \bs S_2)$. Combining this characterization with the linear
regression model in~\eqref{eq:linear-reg}, we arrive at the
following convex covariance-fitting problem for nonnegative APS
recovery.

\noindent
\textbf{Optimization Problem 1 (P-1):}
\begin{equation}
\begin{aligned}
  \underset{ \bs a \in \mathbb{R}^{p+1},\, \bs S_1, \bs S_2 \succeq \bs{0}}{\text{minimize}}
  & \quad \frac{1}{2}\big\|\bs{\Phi}\bs{\Pi} \bs a - \bs{y}\big\|_2^2 \\
  \text{subject to}
  & \quad \bs a = \beta( \bs S_1, \bs S_2).
\end{aligned}
\label{eq:constrainedOpt-nonneg}
\end{equation}
Problem~\eqref{eq:constrainedOpt-nonneg} is a convex SDP with
linear matrix inequality constraints in $(\bs S_1, \bs S_2)$ and a quadratic objective in $\bs a$, and can be handled efficiently by standard SDP
solvers.

\subsection{Connection to EAPM Algorithm}
\label{subsec:eapm}

The EAPM method of~\cite{Miretti-2021} can also be cast in the same
transformed-domain formulation $g(x) = \rho(\arcsin x) \in L^2_w(I)$.
Let $\mathcal{V}_{\bs r} \subset L^2_w(I)$ denote the affine set defined
by the covariance constraints (cf.~Lemma~\ref{lem:Affine-sub-space}),
and let
\(
\mathcal{Z} := \{ f \in L^2_w(I) \mid f(x) \ge 0,\ x \in I \}
\)
denote the nonnegative cone. 
{In this unified transformed domain, the
EAPM algorithm~\cite[Sec.~III-B]{Miretti-2021} aims to find a feasible solution $g_{\mathrm{e}}$, i.e.,}
\begin{equation}
	\text{find } g_{\mathrm{e}} \in L^2_w(I)
	\quad\text{s.t.}\quad
	g_{\mathrm{e}} \in \mathcal{V}_{\bs r} \cap \mathcal{Z}.
	\label{eq:EAPM-feas}
\end{equation}
Algorithmically, this feasibility problem is addressed by an
extrapolated alternating-projection iteration of the form
\begin{equation}
	g^{(t+1)}_{\mathrm{e}}
	= g^{(t)}_{\mathrm{e}} + \mu_t \big[
	\mathbb{P}_{\mathcal{V}_{\bs r}}\big(\mathbb{P}_{\mathcal{Z}}(g^{(t)}_{\mathrm{e}})\big)
	- g^{(t)}_{\mathrm{e}}
	\big],
	\label{eq:EAPM-iter}
\end{equation}
where $\mathbb{P}_{\mathcal{Z}}$ and $\mathbb{P}_{\mathcal{V}_{\bs r}}$ denote the projections onto
$\mathcal{Z}$ and $\mathcal{V}_{\bs r}$, respectively, and $\mu_t$ is an
extrapolation factor chosen according to~\cite{Bauschke2006Extrapolation,
	Miretti-2021,Miretti-2018}. The projection onto $\mathcal{V}_{\bs r}$
admits a closed-form expression via the PLV characterization in
Proposition~\ref{prop:PLV-equivalent}, so that
\eqref{eq:EAPM-iter} can be implemented explicitly in the space
$L^2_w(I)$.


\begin{remark}[EAPM versus Chebyshev--SDP formulation]
\label{rem:EAPM-vs-cheb}
The EAPM algorithm and the proposed Chebyshev–SDP framework enforce
the same physical constraints, i.e., APS
nonnegativity, but from different viewpoints.
\begin{itemize}
    \item Within the unified transformed-domain
framework, we highlight the conceptual difference between EAPM,
which treats nonnegative APS recovery as an \emph{infinite-dimensional}
feasibility problem solved by alternating projections, and the
Chebyshev-based approach developed in Section~\ref{subsec:SDP}, which operates with a \emph{finite-dimensional} parametrization and a convex optimization formulation.
\item 
The numerical implementation of EAPM necessarily
relies on a finite discretization of the angular domain.
As a consequence, both nonnegativity and covariance consistency are
enforced only up to the accuracy of the underlying grid
resolution, and the algorithmic behavior can be sensitive to the choice
of discretization and step-size sequence $\{\mu_t\}$ in
\eqref{eq:EAPM-iter}.
\end{itemize}

\end{remark}

\subsection{Weighted Multi-Order Derivative Regularization}

Under limited aperture ($M<\infty$), APS recovery is typically ill-posed: Many data-consistent solutions exist and mainly differ in spurious angular ripples. We target a {clustered, piecewise-smooth} APS---smooth within clusters and with a few transition angles---consistent with clustered APS models~\cite{zhong2020fdd,song2020deep}.

This prior encourages angular clusters to follow a low-order (piecewise-polynomial) profile, while concentrating the regularization cost on transition angles.
This interpretation is consistent with higher-order total variation (TV), most notably total generalized variation (TGV): The order-$k$ regularization is explicitly built to control the $k$-th order derivative, and vanishes on polynomials of degree $<k$~\cite[Proposition~3.3]{bredies2010total}. Moreover, this regularization is governed mainly by derivative jumps across transitions~\cite[Remark~3.7]{bredies2010total}.

Motivated by the above, we regularize the sampled APS via an analysis-type prior~\cite{Elad2007}.
Let $\bs\rho\in\mathbb{R}^K$ denote $K$ APS samples on a uniform grid with spacing $\Delta$, and let $\bs D$ be a first-order (non-circulant) difference operator ($\bs{D}$ is lower-triangular, refer to~\cite{Interpretable-FLR-2009}). Define the $n$-th order discrete derivative proxy
\begin{equation}
	\bs\rho^{(n)} := \Delta^{-n}\bs D^n\bs\rho,\qquad n\ge 0,
\end{equation}
and the weighted multi-order response
\begin{equation}
	\label{eq:Q_def_concise}
	\bs\xi := \sum_{n=0}^{K}\eta_n\bs\rho^{(n)} \;=\; \bs Q_{\eta}\bs\rho,
	\quad
	\bs Q_{\eta} := \sum_{n=0}^{K}\eta_n\Delta^{-n}\bs D^n ,
\end{equation}
with $\bs D^{K+1}=\bs 0$. We penalize $\|\bs\xi\|_1$ to promote sparsity of the detected variations (cf.\ $\ell_1$ trend filtering in~\cite{kim2009ell_1}), yielding smooth interiors while preserving a few transitions~\cite{Stefan2010}.

We adopt the geometric weighting $\eta_n = \eta^{n-n_0}$ with $\eta >0$ for $n \geq n_0$, and $\eta_n=0$ otherwise,
which emphasizes orders $\ge n_0$. Setting $n_0=2$ reduces sensitivity to the overall power level and alleviates staircasing by promoting low-order interiors~\cite{chan2000high}.

Using the nilpotency of $\bs D$ and the finite Neumann series~\cite{HornJohnson2013,GolubVanLoan2013}, \eqref{eq:Q_def_concise} implies
\begin{equation}
	\label{eq:Q_neumann_concise}
	\bs Q_{\eta}
	=
	\Delta^{-n_0}\bs D^{n_0}\big(\bs I-\tfrac{\eta}{\Delta}\bs D\big)^{-1},
\end{equation}
The matrix $\bs{Q}_{\eta}$ is effectively banded for a proper ratio $\eta/\Delta$.

\begin{remark}
	\textit{(FIR and sparsity viewpoint.)}
	Because $\bs Q_{\eta}$ is effectively banded, $\bs\xi=\bs Q_{\eta}\bs\rho$ acts as a window-like FIR high-pass detector along the grid. Penalizing $\|\bs\xi\|_1$ suppresses high-frequency ripples at most angles (smooth interiors) while allowing a few localized large responses (cluster transitions). Emphasizing higher-order terms increases sensitivity to curvature-like changes and promotes low-order interiors, consistent with higher-order TV effects~\cite{chan2000high,kim2009ell_1} in image processing.
\end{remark}

Combining this analysis prior with the nonnegative constraint in Lemma~\ref{lem:sdp_cheb} yields:

\noindent\textbf{Optimization Problem 2 (P-2):}
\begin{equation}
	\label{eq:P2_clust_concise}
	\begin{aligned}
		\underset{\bs a,\bs\xi,\bs S_1,\bs S_2}{\text{minimize}}
		&\quad
		\frac{1}{2}\big\|\bs{\Phi}\bs{\Pi}\bs a-\mathbf{y}\big\|_2^2
		+\lambda\|\bs\xi\|_1\\
		\text{subject to}
		&\quad \bs a=\beta(\bs S_1,\bs S_2),\ \bs S_1,\bs S_2\succeq \bs 0,\\
		&\quad \bs\xi=\bs Q_{\eta}\bs C\bs a,
	\end{aligned}
\end{equation}
where $\bs C\in\mathbb{R}^{K\times(p+1)}$ evaluates the Chebyshev basis at $\{\sin\theta_k\}_{k=0}^{K-1}$ (cf.~\eqref{eq:rho-expansion}) with
\[
[\bs C]_{k,n} = \frac{1}{\sqrt{1 + \delta_{n0}}} T_n(\sin \theta_k), \quad 0 \leq n \leq p,
\]
where $\{\theta_k\}_{k=0}^{K-1}$ are the uniform girds, so that $\bs\rho=\bs C\bs a$. Here $\lambda\ge 0$ controls the tradeoff between data consistency and structural regularity.







\begin{figure*}[htbp]
	\centering
	\subfloat[Multi-Gaussian]{\includegraphics[width=0.31\textwidth]{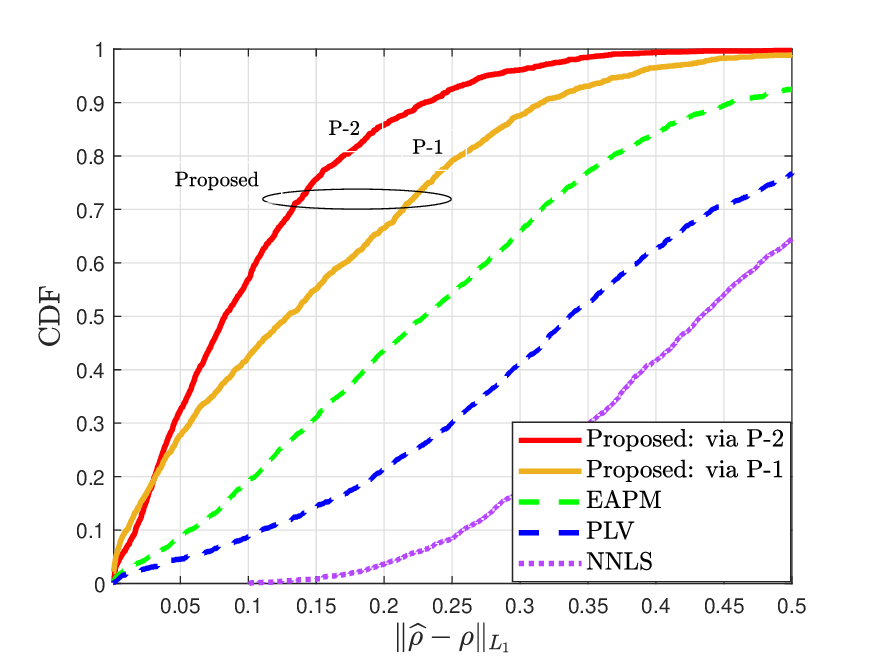}}
	\hfil
	\subfloat[{Multi-Sinc$^2$}]{\includegraphics[width=0.31\textwidth]{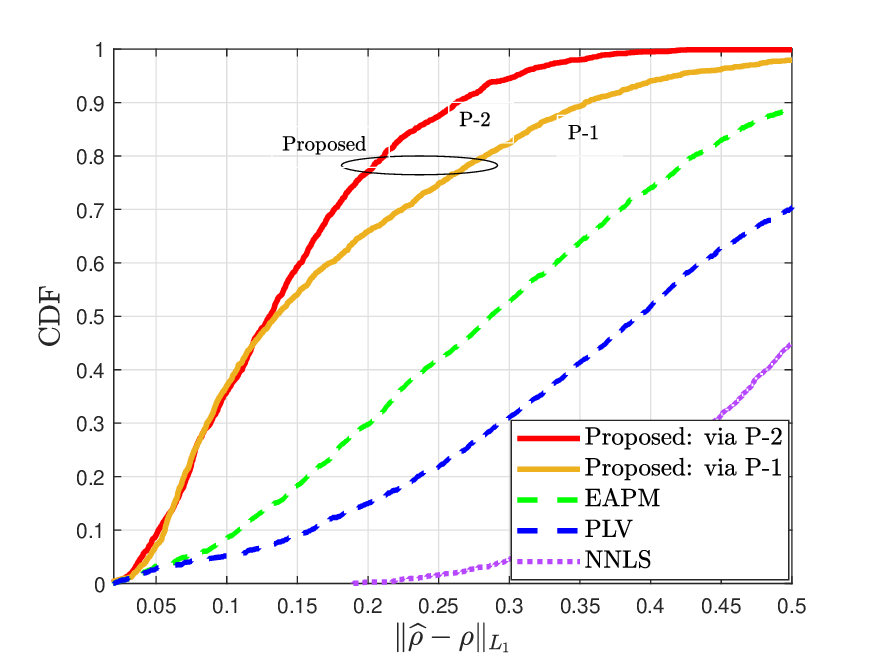}}
	\hfil
	\subfloat[{Multi-Laplacian}]{\includegraphics[width=0.31\textwidth]{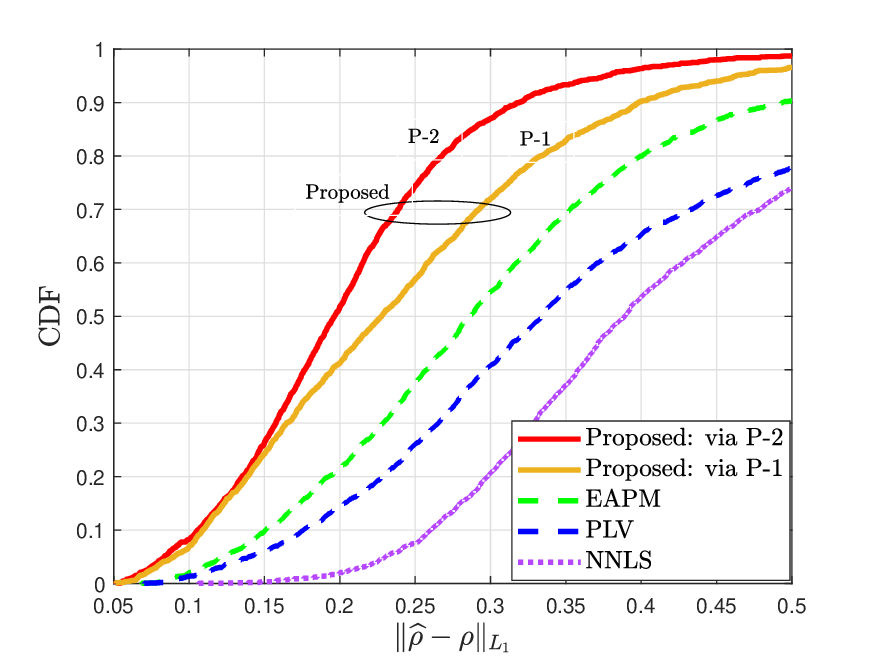}}
	\caption{The $L_1$ distortion performance comparisons of different continuous APS estimators under three APS models. }
	\label{fig:my_figure1}
\end{figure*}

\begin{figure*}[htbp]
	\centering
	\subfloat[Multi-Gaussian]{\includegraphics[width=0.31\textwidth]{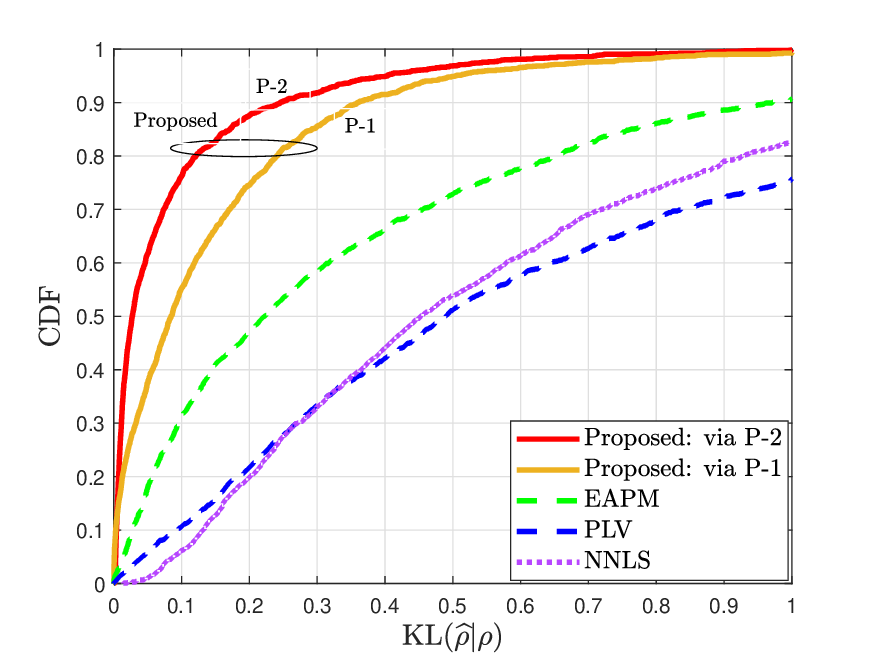}}
	\hfil
	\subfloat[{Multi-Sinc$^2$}]{\includegraphics[width=0.31\textwidth]{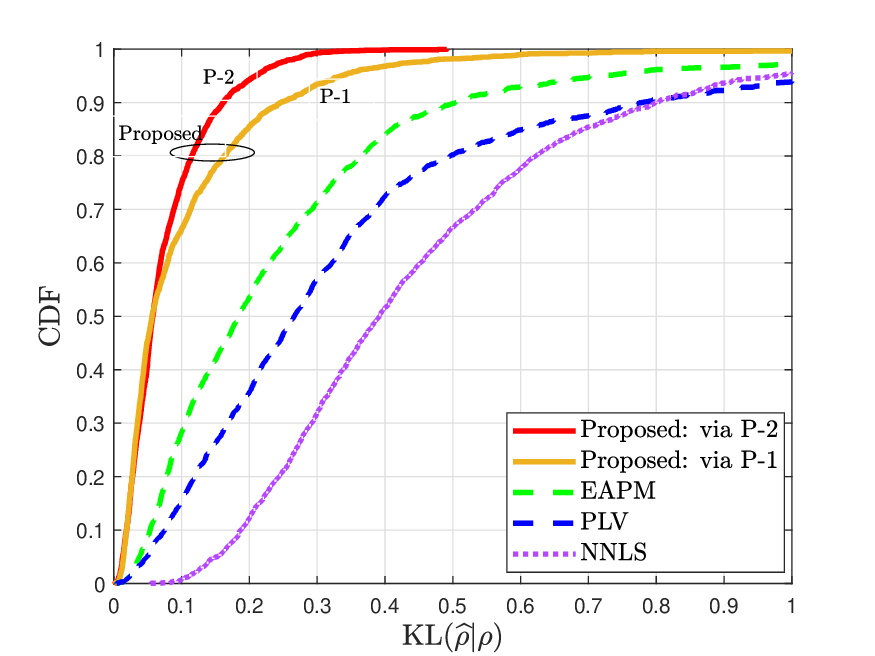}}
	\hfil
	\subfloat[{Multi-Laplacian}]{\includegraphics[width=0.31\textwidth]{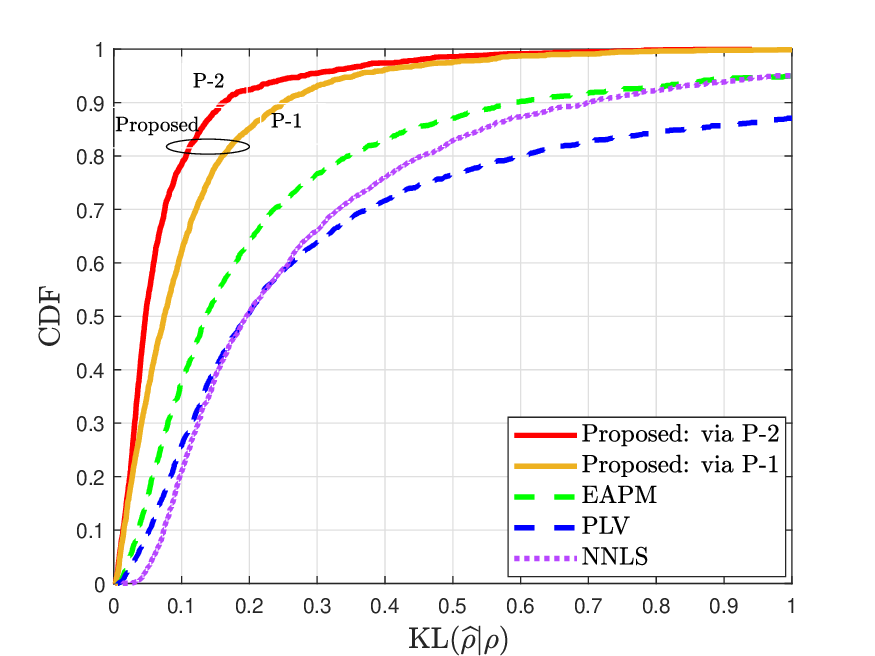}}
	\caption{The perceptional performance comparisons (via KL divergence) of different continuous APS estimators. }
	\label{fig:my_figure2}
\end{figure*}

\section{Simulations}
\label{Sec:simulation}

This section shows the performance of the proposed
Chebyshev-based APS recovery framework through numerical
simulations. We first consider synthetic clustered APS models to
quantify reconstruction accuracy from a finite number of covariance
coefficients and to compare against some baselines.
We then examine robustness to within-cluster perturbations and assess
the impact on FDD downlink covariance conversion in a 3GPP-based
macro-cell scenario~\cite{3GPP2024TR38901}. 

\subsection{APS Recovery Under Clustered Models}

We first examine APS recovery in controlled settings where the
true APS is generated from clustered analytical models. Following~\cite{Miretti-2021,Bameri-2023,Ninomiya-2022,Kaneko-2024}, we represent
\(
  \rho(\theta)
  =
  \sum_{q=1}^{Q} c_q f_q(\theta),
\)
where $Q$ is the number of clusters, $c_q > 0$ denotes the power
of the $q$-th cluster, and $f_q(\theta)$ models its angular shape/location. Consider three clustered APS
models:
\begin{itemize}
  \item \emph{Gaussian}:
        $f_q(\theta) = \exp\!(-(\theta-\mu_q)^2/(2\sigma_q^2))$,
        where $\mu_q$ and $\sigma_q$ denote the cluster center and
        angular spread~\cite{Miretti-2021,Bameri-2023}.
  \item \emph{Sinc$^2$}:
        $f_q(\theta) = \mathrm{sinc}^2\!((\theta-\mu_q)/(2\sigma_q))$,
        where $\mu_q$ is the center, and $\sigma_q$ controls the main-lobe width and the function
        exhibits rapidly decaying side-lobes.
  \item \emph{Laplacian}:
        $f_q(\theta) = \exp\!(-|\theta-\mu_q|/\sigma_q)$,
        with center $\mu_q$ and spread parameter $\sigma_q$; see also
        \cite[Sec.~IV-A]{Kaneko-2024}.
\end{itemize}

\emph{Basic setup.} 
 We set $\gamma=1$ (corresponding to
half-wavelength antenna spacing \cite{Miretti-2021,Bameri-2023}),
the number of antennas $M=8$. The APS is normalized to have unit mass (see~\cite{Bameri-2023}),
\(
  \int_{\Theta} \rho(\theta)\,\mathrm{d}\theta = 1,
\)
which is implemented by normalizing, i.e., $r_m \leftarrow r_m/r_0$.

\begin{table}[t]
\caption{Parameter settings for angular centers and spreads (angles in degrees, then converted to radians in simulations).}
\label{Tab:angular-setting}
\centering
\renewcommand{\arraystretch}{1.2}
\begin{tabular}{l p{0.36\columnwidth} p{0.29\columnwidth}}
\toprule
Cluster type & Cluster center distribution & Spread distribution \\
\midrule
Multi-Gaussian  & $\mu_q \sim \mathcal{U}(-78^\circ, 78^\circ)$ & $\sigma_q \sim \mathcal{U}(3.0^\circ, 9.0^\circ)$ \\
Multi-Sinc$^2$  & $\mu_q \sim \mathcal{U}(-78^\circ, 78^\circ)$ & $\sigma_q \sim \mathcal{U}(4.8^\circ, 11.2^\circ)$ \\
Multi-Laplacian & $\mu_q \sim \mathcal{U}(-78^\circ, 78^\circ)$ & $\sigma_q \sim \mathcal{U}(2.1^\circ, 10.6^\circ)$ \\
\bottomrule
\end{tabular}
\end{table}

\begin{figure*}[htbp]
	\centering
	\subfloat[Multi-Gaussian]{\includegraphics[width=0.31\textwidth]{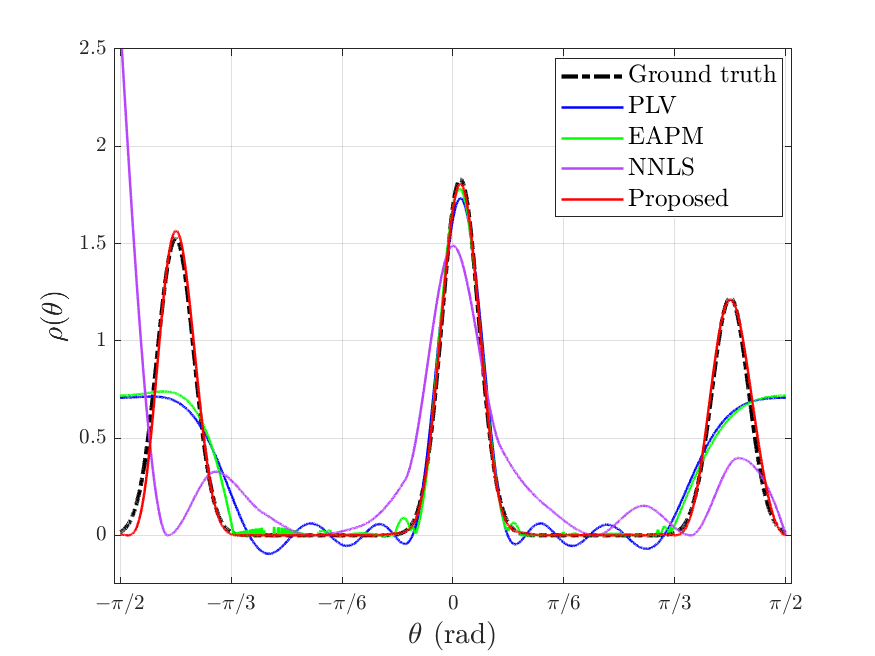}}
	\hfil
	\subfloat[{Multi-Sinc$^2$}]{\includegraphics[width=0.31\textwidth]{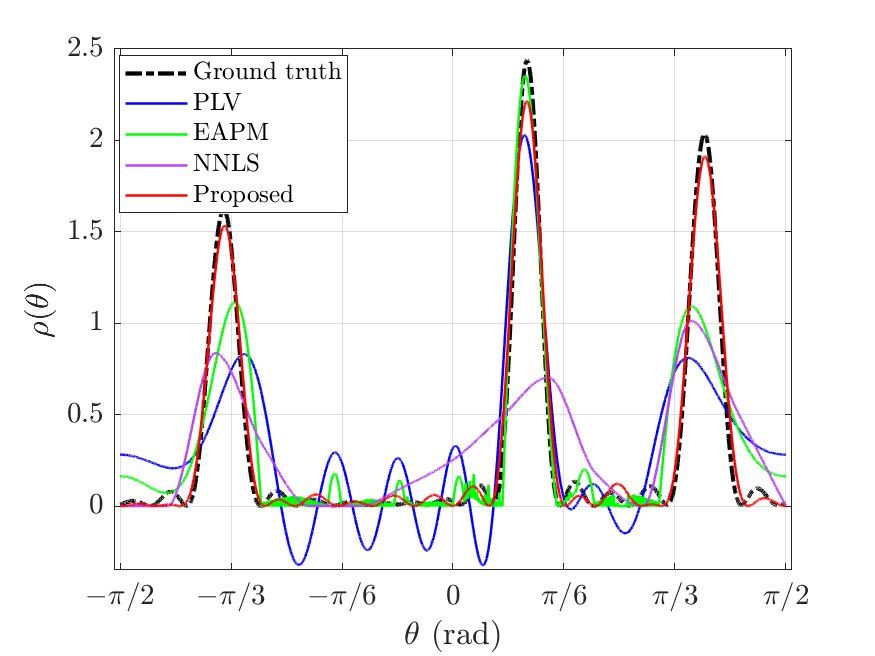}}
	\hfil
	\subfloat[{Multi-Laplacian}]{\includegraphics[width=0.31\textwidth]{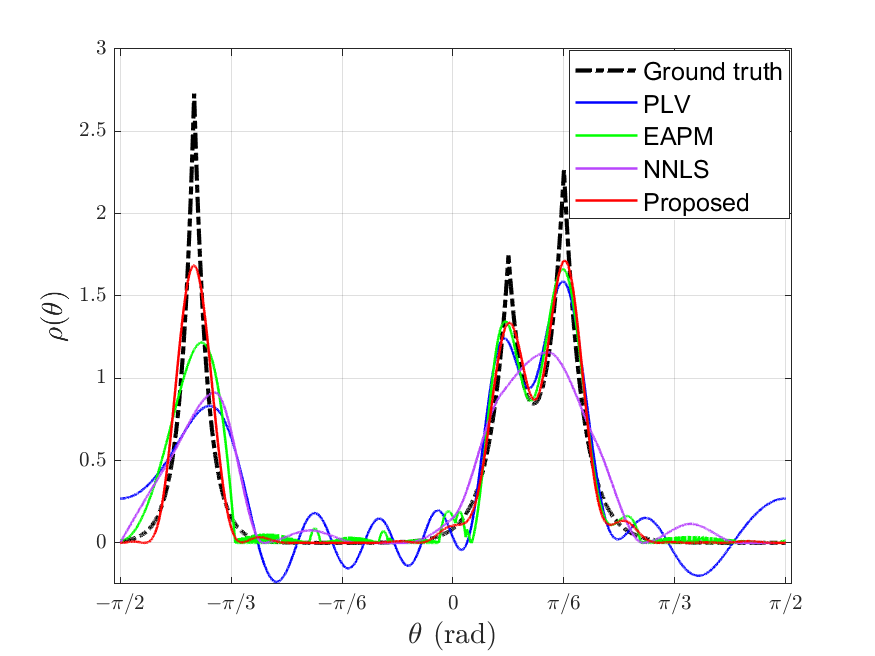}}
	\caption{The visual performance comparisons of different continuous APS estimators under three APS models.}
\label{fig:my_figure}
\end{figure*}

\begin{figure*}[htbp]
	\centering
	\subfloat[Multi-Gaussian]{\includegraphics[width=0.31\textwidth]{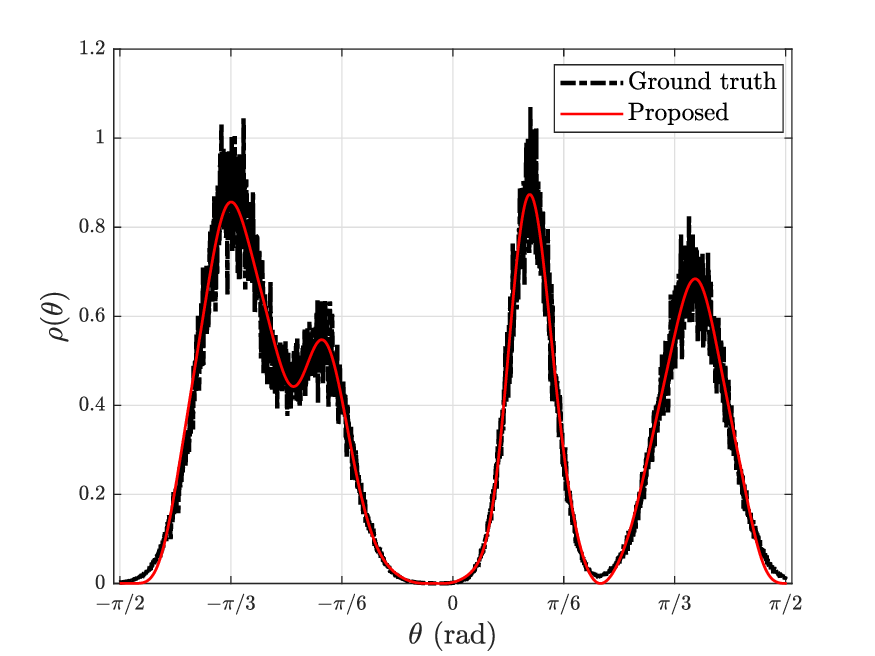}}
	\hfil
	\subfloat[{Multi-Sinc$^2$}]{\includegraphics[width=0.31\textwidth]{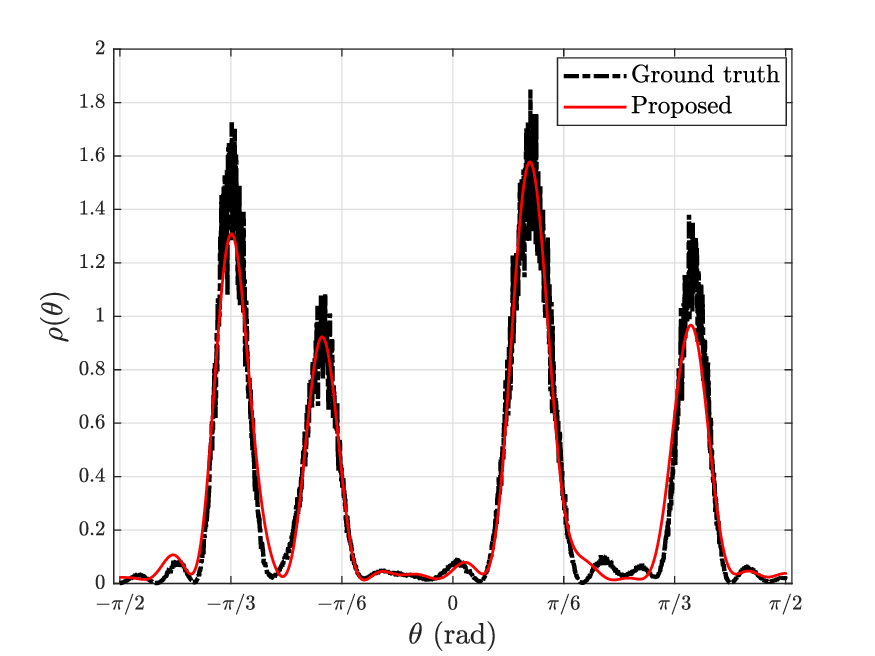}}
	\hfil
	\subfloat[{Multi-Laplacian}]{\includegraphics[width=0.31\textwidth]{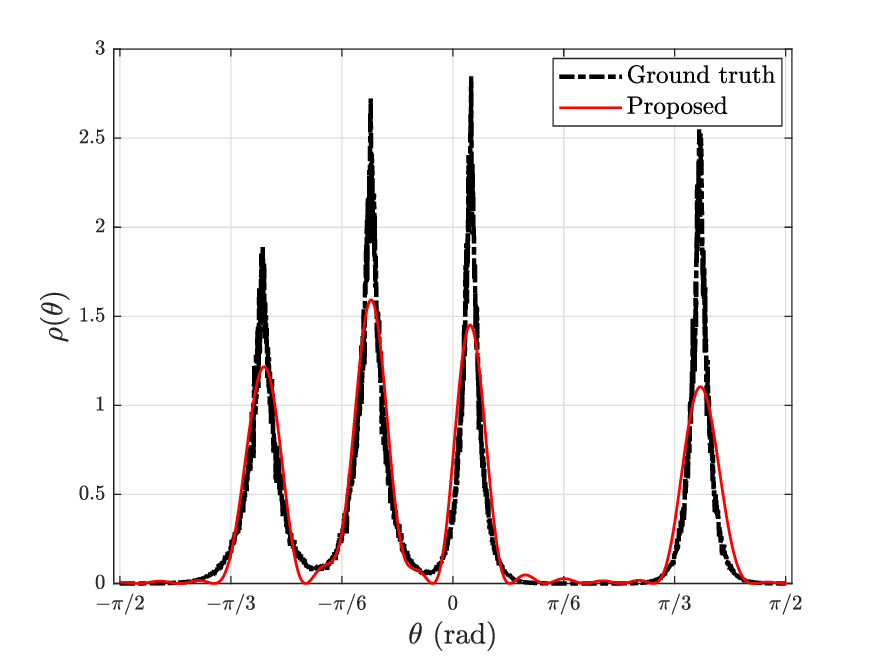}}
	\caption{The visual performance comparisons of different continuous APS models with perturbation level $\sigma = 0.2$. }
	\label{fig:my_figure3}
\end{figure*}

We evaluate APS recovery by three complementary criteria:
\begin{enumerate}
  \item \emph{$L_1$ distortion:} 
        $\|\hat\rho - \rho\|_{L_1}
        := \int_{\Theta} |\hat\rho(\theta)-\rho(\theta)|\,\mathrm{d}\theta$.
  \item \emph{Perceptual loss:} The Kullback--Leibler (KL) divergence
        \(
          \mathrm{KL}(\hat\rho | \rho)
          :=
          \int_{\Theta}
            \hat\rho(\theta)
            \log\! ({\hat\rho(\theta)}/{\rho(\theta)})
          \,\mathrm{d}\theta,
          \)
        where the APS is normalized.
        Since the APS is physically
        nonnegative and can be viewed as a probability density on
        $\Theta$, KL divergence provides a perceptional measure of how well the recovered APS matches the ground
        truth~\cite{zhang2019theoretically,Qian-WeiYu-2025}.
    \item \emph{Visual comparison:} Representative realizations with
    $Q=3$ clusters for different APS models.
\end{enumerate}

\emph{Monte Carlo settings.}
We perform $2000$ Monte-Carlo trials. In each trial, the number of clusters $Q$ is uniformly drawn from $\{1,2,3,4,5\}$, the cluster powers
$c_q \sim \mathcal{U}(0.5,1)$, and the angular centers and spreads
$(\mu_q,\sigma_q)$ are randomly generated based on the ranges
outlined in Table~\ref{Tab:angular-setting} for each APS model. For each realization, we generate the covariance and apply:
\begin{itemize}
  \item The Proposed schemes, i.e., \textbf{P-1} in~\eqref{eq:constrainedOpt-nonneg} and \textbf{P-2} in~\eqref{eq:P2_clust_concise}. Set the model order $p=31$ and $\eta/\Delta=0.2$;
  \item The PLV solution and its EAPM refinement
        \cite{Miretti-2021,Miretti-2018};
  \item The NNLS-based estimator \cite{Haghighatshoar-2018,Barzegar-2019}
        with grid size $2M-1$.
\end{itemize}
For NNLS, the discrete APS estimate on the angular grid is interpolated
to a continuous function to enable a fair comparison with the
continuous-domain methods, following \cite{akima1970new}.

Figs.~\ref{fig:my_figure1} and Fig.~\ref{fig:my_figure2} show that the proposed Chebyshev-based estimator achieves the lowest $\ell_1$ distortion and KL divergence among the compared methods. This indicates that jointly enforcing APS smoothness and nonnegativity within the Chebyshev parametrization yields more accurate APS reconstructions from a limited number of covariance coefficients. 
{In particular, \textbf{P-2} provides the best overall performance and is therefore used in the subsequent experiments.}

Fig.~\ref{fig:my_figure} shows that the proposed Chebyshev-based approach yields APS estimates that closely match the ground-truth clustered profiles across all tested models. By contrast, PLV/EAPM introduce oscillatory artifacts and spurious sidelobes, particularly when clusters are near the angular boundaries, which is consistent with their trigonometric-polynomial parametrization.


\subsection{Robustness to Perturbed Clustered APS}

To assess robustness beyond idealized clustered models, we consider perturbed APS realizations obtained by injecting controlled intra-cluster fluctuations while preserving the dominant cluster supports. Specifically, starting from an analytical clustered profile $\rho(\theta)$, we define
\begin{equation}
	\rho_{\mathrm{pert}}(\theta)
	= \rho(\theta)\,\big(1 + \sigma \mathcal{E}_{\theta}\big)_{+},
	\label{eq:perturbed-APS}
\end{equation}
where $(\cdot)_{+}$ enforces nonnegativity, $\sigma>0$ controls the perturbation level, and $\mathcal{E}_{\theta}$ is a zero-mean, unit-variance correlated process formed by smoothing i.i.d.\ Gaussian samples with a short averaging window (length $6$). 
{This construction retains a clustered support while introducing within-cluster irregularities, serving as a controlled model-mismatch stress test commonly used in robustness evaluations~\cite{chi2011sensitivity,wang2022policy,liu2015robust}.}

\begin{figure}[t]
\centering
\includegraphics[width=\linewidth]{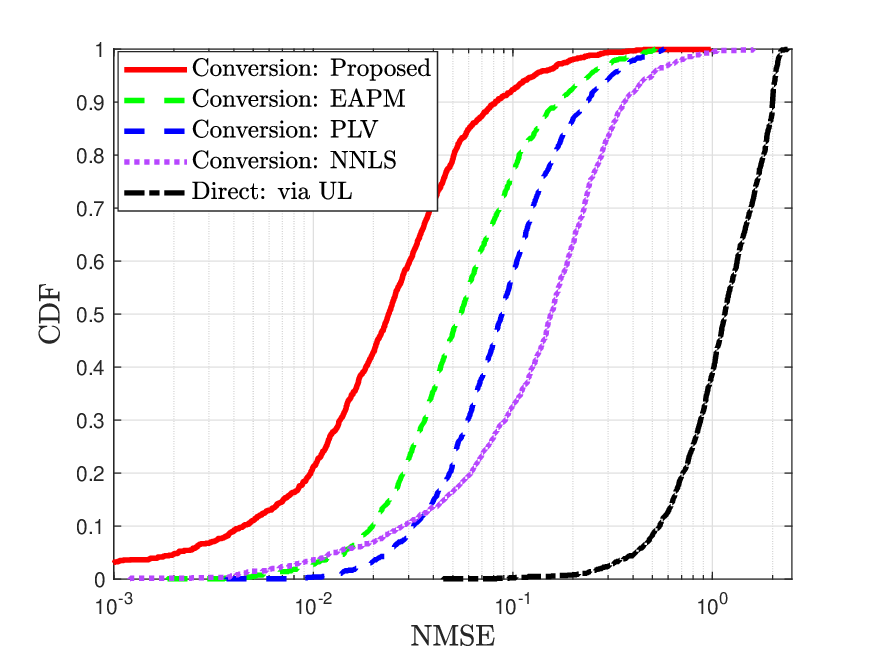}
\caption{Downlink channel covariance conversion NMSE performance in the 3GPP\_38.901\_UMa\_NLOS FDD scenario.}
\label{fig:my_figure-FDD}
\end{figure}

Fig.~\ref{fig:my_figure3} illustrates a representative case with perturbation level
$\sigma = 0.2$ and $M=16$. The proposed Chebyshev-based scheme continues to
recover the main cluster locations and their relative power with high
fidelity, despite the visibly irregular fine-scale variations within
each cluster. These experiments indicate that the combination of
global nonnegativity, smoothness-aware Chebyshev modeling, and
derivative-based regularization yields APS reconstructions that remain
stable under moderate deviations from the ideal clustered model.

\subsection{Application: DL Channel Covariance Conversion in FDD}

We assess the impact of the proposed APS scheme on
DL covariance conversion in a realistic FDD setting. Channel
realizations are generated using the QuaDRiGa toolbox under the
3GPP\_38.901\_UMa\_NLOS scenario~\cite{jaeckel2014quadriga,Jaeckel2023QuaDRiGa,3GPP2024TR38901},
with a BS equipped with an $M=8$ ULA (25\,m height). Users
are randomly distributed over $[-90^\circ,90^\circ]$ in azimuth and
uniformly in distance between 150\,m and 200\,m from the BS. To ensure
the same scattering environments on the UL and DL, we use the same
random seed for both links, so that the clusters (powers, AoAs,
delays) are shared while the carrier frequencies differ
($f_{\mathrm{UL}} = 1.7$\,GHz, $f_{\mathrm{DL}} = 2.5$\,GHz). This
yields distinct UL/DL covariance matrices corresponding to the same
underlying APS.

For each user, the empirical channel covariance matrices are computed from $N_{\text{snap}} = 128$ channel snapshots. The
competing methods are allowed to use only the UL sample covariance to
predict the DL covariance. The proposed Chebyshev–SDP estimator first
recovers the Chebyshev series from the UL covariance and then maps it to the DL covariance via Theorem~\ref{theorem:CB-parseval}. Following~\cite{Miretti-2021}, performance is evaluated using the normalized mean-squared error (NMSE) between the converted and true DL covariance matrices, averaged over 2000 users.

Fig.~\ref{fig:my_figure-FDD} shows that the proposed estimator
consistently attains lower NMSE 
than the baselines. In particular,
despite the non-ideal scattering captured by the UMa\_NLOS model, the
Chebyshev-based conversion leads to a more accurate prediction of the DL covariance.


\section{Conclusion}
\label{sec:conclusion}

This paper developed a Chebyshev-based framework for continuous APS
recovery from channel covariance. By working in a transformed
domain and expanding the APS in a Chebyshev basis on a finite interval,
we obtained an explicit series representation of the covariance
entries and a truncated linear regression formulation whose
approximation error is directly tied to the decay of the Chebyshev
coefficients, and hence to the smoothness of the APS. Based on this
representation, we derived an exact semidefinite characterization of
nonnegative APS and  incorporated a derivative-based $\ell_{1}$ regularizer to suppress spurious angular ripples while preserving clustered transitions.
Simulations with synthetic clustered spectra and a 3GPP-based FDD setting demonstrate improved APS reconstruction and downlink covariance conversion compared with existing baselines, highlighting the utility of the proposed approach for covariance-domain processing in multi-antenna systems.

\appendices

\section{Proof of Theorem~\ref{theorem:CB-parseval}}
\label{appendix-pf-theorem:CB-parseval}

The following proposition will be useful for the proof. 
\begin{proposition}[Parseval Identity for Chebyshev Expansion]
	\label{prop:parseval}
	~\cite[Theorem 6.10]{Plonka2023Chebyshev}
	Let $g\in L^2_w$ be Lipschitz continuous with Chebyshev expansion $g(x)=\sum_{n=0}^\infty a_n T_n(x)$.
	Then, the following holds:
	\begin{align}
		\Vert g \Vert^2_w
		= \pi |a_0|^2 + \frac{\pi}{2}\sum_{n=1}^\infty |a_n|^2.
		\label{eq:parseval}
	\end{align}
\end{proposition}

\textbf{Proof of (i)}. 
We expand the exponential term $e^{\ii\kappa x}$ in \eqref{eq:r-fun} via the Chebyshev basis $\{T_n\}$ and compute the related coefficients (Chebyshev series). Define
\[
\begin{split}
	b_n(\kappa)
	&:=
	\frac{2-\delta_{n0}}{\pi}\int_{-1}^{1} e^{\ii\kappa x}\,T_n(x)\,w(x)\,\dd x.
\end{split}
\]
Apply the substitution $x=\cos\varphi$ with $\varphi\in[0,\pi]$; then $w(x)\,\dd x=\dd\varphi$ and $T_n(x)=\cos(n\varphi)$. Hence
\[
\begin{split}
	b_n(\kappa)
	&=\frac{2-\delta_{n0}}{\pi}\int_{0}^{\pi} e^{\ii\kappa\cos\varphi}\,\cos(n\varphi)\,\dd\varphi.
\end{split}
\]
The classical integral identity (see \cite[Eq.(10.9.2)]{NISTDLMF}) implies
\begin{align}
	\int_{0}^{\pi} e^{\ii\kappa\cos\varphi}\,\cos(n\varphi)\,\dd\varphi
	=\pi\,\ii^{\,n}\,J_n(\kappa),
	\ \ n\in\mathbb{N},
	\label{eq:bessel-cos-int}
\end{align}
where $J_n(\cdot)$ is the Bessel function of the first kind of order $n$. It
immediately yields $b_0(\kappa)=J_0(\kappa)$ and $b_n(\kappa)=2\,\ii^{\,n}\,J_n(\kappa) ,n\ge 1$.
Therefore, for every fixed $\kappa\in\mathbb{R}$,
\begin{align}
	\label{eq:plane-cheb-series}
	e^{\ii\kappa x}
	=
	J_0(\kappa)
	+
	2\sum_{n=1}^{\infty} \ii^{\,n}J_n(\kappa)\,T_n(x),
	\qquad x \in I.
\end{align}
Substituting $g=\sum_{n\ge 0}a_nT_n$ and \eqref{eq:plane-cheb-series} into \eqref{eq:r-fun}, we have 
\[
\begin{aligned}
	& r (\kappa)  \\
	=&
	\int_{I}
	\Big(\sum_{j=0}^{\infty} a_j T_j(x)\Big)
	\Big(J_0(\kappa)+2\sum_{n=1}^{\infty}\ii^{\,n}J_{n}(\kappa)T_{n}(x)\Big)
	w(x)\,\dd x\\
	\overset{(a)}{=}&a_0\,J_0(\kappa)\,\underbrace{\langle T_0, T_0 \rangle_{w}}_{=\pi}
	+
	\sum_{n=1}^{\infty} a_n\,(2\ii^{\,n}J_n(\kappa))\,\underbrace{\langle T_n, T_n \rangle_{w}}_{=\pi/2}
	\notag\\
	&=
	\pi\Big[a_0\,J_0(\kappa)+\sum_{n=1}^{\infty}\ii^{\,n}a_n J_n(\kappa)\Big],\quad\forall \kappa\in\mathbb{R},
\end{aligned}
\]
where step $(a)$ follows from the orthogonality of Chebyshev polynomials \eqref{eq:OrthoPoly}, (hence only terms with matching indices survive), and the interchange of summation and integration. (This is justified by the absolute convergence of the series, which is further due to:
$\sum_{n\ge 0}|a_n|^2<\infty$ (by $g\in L^2_w$) and
the Bessel energy identity \cite[§13.31]{WatsonBessel}):
\begin{equation}
	\label{eq:Bessel-Identity}
	J_0(\kappa)^2+2\sum_{n\ge 1}J_n(\kappa)^2=1.
\end{equation}

\textbf{Proof of (ii)}. 
Using the same reasoning as in (i), we can obtain $r_{\leq p}(\kappa)=\pi\sum_{n=0}^p \mathrm i^n a_n J_n(\kappa)$; hence $r(\kappa)-r_{\leq p}(\kappa)=\sum_{n>p} \mathrm i^n a_n J_n(\kappa)$. Accordingly,
\begin{equation}
	\label{eq:ep-derivation}
	\begin{split}
		|r(\kappa)-r_{\leq p}(\kappa)|
		&\le
		\pi\sum_{n>p}|a_n||J_n(\kappa)|
		\\
		&\le \pi
		(\sum_{n>p}|a_n|^2)^{1/2} \cdot (\sum_{n>p}J_n(\kappa)^2)^{1/2}
		\\
		& \le \frac{\pi}{\sqrt{2}}
		(\sum_{n>p}|a_n|^2)^{1/2}
	\end{split}
\end{equation}
where the second inequality uses the Cauchy--Schwarz inequality, and the last inequality is from \eqref{eq:Bessel-Identity}.

Further, for the residual, Proposition~\ref{prop:parseval} gives 
\[
\|g-g_{\leq p}\|_{w}^2
= \frac{\pi}{2}\sum_{n>p}|a_n|^2.
\]
Substituting this result into~\eqref{eq:ep-derivation} yields
\[
|r(\kappa) - r_{\le p}(\kappa)|
\le \sqrt{\pi}\, \|g - g_{\le p}\|_{w},
\]
which corresponds to~\eqref{eq:tail-function}. This completes the proof.
\hfill$\blacksquare$

\section{Proof of Lemma~\ref{lem:Affine-sub-space}}
\label{proof:affine}
	\noindent
		\text{1) Affine structure.}
		Let $f_1\in \mathcal V_r$ and $f_2\in \mathcal N$. For every $-M<m<M$, by linearity of the inner product in the first argument,
		\[
		\langle f_1+f_2, e^{i\kappa_m(\cdot)}\rangle_w
		=\langle f_1, e^{i\kappa_m(\cdot)}\rangle_w+\langle f_2, e^{i\kappa_m(\cdot)}\rangle_w
		=r_m,
		\]
		so $f_1+f_2\in \mathcal V_r$. Hence $\mathcal V_r+\mathcal N\subseteq \mathcal V_r$.
		Conversely, fix any $f_1\in \mathcal V_r$ and take any $f\in \mathcal V_r$. Then for each $-M<m<M$,
		\[
		\langle f-f_1, e^{i\kappa_m(\cdot)}\rangle_w
		=\langle f, e^{i\kappa_m(\cdot)}\rangle_w-\langle f_1, e^{i\kappa_m(\cdot)}\rangle_w
		=0,
		\]
		so $f-f_1\in \mathcal N$, i.e., $f\in f_1+\mathcal N\subseteq \mathcal V_r+\mathcal N$.
		Therefore $\mathcal V_r\subseteq \mathcal V_r+\mathcal N$, and we conclude
		\[
		\mathcal V_r+\mathcal N=\mathcal V_r,
		\]
		so $\mathcal V_r$ is an affine subset of $L^2_w(I)$ with direction $\mathcal N$.
		
		\noindent
		\text{2) Characterization of $\mathcal N_\perp$.}
		Define the finite-dimensional subspace
		\[
		\mathcal S:=\operatorname{span}\bigl(\,1,\ \{\cos(\kappa_m \,\cdot)\}_{m=1}^{M-1},\ \{\sin(\kappa_m \,\cdot)\}_{m=1}^{M-1}\bigr)\subset L^2_w(I).
		\]
		We claim that $\mathcal N=\mathcal S^\perp$. Indeed, if $f\in \mathcal N$, then in particular
		\[
		\langle f,1\rangle_w=\langle f,e^{i\kappa_0(\cdot)}\rangle_w=0,
		\]
		and for each $m=1,\dots,M-1$ we have $\langle f,e^{i\kappa_m(\cdot)}\rangle_w=0$ and
		$\langle f,e^{-i\kappa_m(\cdot)}\rangle_w=0$ (since the constraints hold for all $-M<m<M$).
		Using the identities
		\[
		\begin{split}
		\cos(\kappa_m x)&=\frac{e^{i\kappa_m x}+e^{-i\kappa_m x}}{2},\\
		\sin(\kappa_m x)&=\frac{e^{i\kappa_m x}-e^{-i\kappa_m x}}{2i},
		\end{split}
		\]
		and linearity of $\langle\cdot,\cdot\rangle_w$ in the second argument, we obtain
		\[
		\langle f,\cos(\kappa_m(\cdot))\rangle_w
		=\frac{1}{2}\big(\langle f,e^{i\kappa_m(\cdot)}\rangle_w+\langle f,e^{-i\kappa_m(\cdot)}\rangle_w\big)=0,
		\]
		\[
		\langle f,\sin(\kappa_m(\cdot))\rangle_w
		=\frac{1}{2i}\big(\langle f,e^{i\kappa_m(\cdot)}\rangle_w-\langle f,e^{-i\kappa_m(\cdot)}\rangle_w\big)=0.
		\]
		Hence $f$ is orthogonal to each generator of $\mathcal S$, i.e., $f\in \mathcal S^\perp$, so $\mathcal N\subseteq \mathcal S^\perp$.
		Conversely, if $f\in \mathcal S^\perp$, then $\langle f,1\rangle_w=0$ and for each $m=1,\dots,M-1$,
		$\langle f,\cos(\kappa_m(\cdot))\rangle_w=0$ and $\langle f,\sin(\kappa_m(\cdot))\rangle_w=0$; therefore
		\[
		\langle f,e^{\pm i\kappa_m(\cdot)}\rangle_w
		=\langle f,\cos(\kappa_m(\cdot))\rangle_w  \pm i\,\langle f,\sin(\kappa_m(\cdot))\rangle_w
		=0,
		\]
		and also $\langle f,e^{i\kappa_0(\cdot)}\rangle_w=\langle f,1\rangle_w=0$. This is exactly the defining condition of $\mathcal N$,
		so $\mathcal S^\perp\subseteq \mathcal N$. Hence $\mathcal N=\mathcal S^\perp$.
		
		Taking orthogonal complements yields
		\[
		\mathcal N_\perp=\mathcal N^\perp=(\mathcal S^\perp)^\perp=\overline{\mathcal S}.
		\]
		Since $\mathcal S$ is finite-dimensional, it is closed in $L^2_w(I)$, so $\overline{\mathcal S}=\mathcal S$ and thus $\mathcal N_\perp=\mathcal S$.
		Equivalently, every $f\in \mathcal N_\perp$ can be written as
		\[
		\begin{split}
			f(x)&=b_0+\sum_{m=1}^{M-1} b_m \cos(\kappa_m x)+ b_{M-1+m}\sin(\kappa_m x)
			\\
			&:=f(x;b),
		\end{split}
		\]
		for some $b\in\mathbb{R}^{2M-1}$, and conversely every such trigonometric polynomial belongs to $\mathcal N_\perp$.
		This proves Lemma~\ref{lem:Affine-sub-space} and the finite-dimensionality claim.
	\hfill$\blacksquare$

\section{Proof of Proposition~\ref{prop:PLV-equivalent}}
\label{proof:proof-PLV}
We show the equivalence of the four statements in
Proposition~\ref{prop:PLV-equivalent} by linking them pairwise.
Throughout the proof, we use the sets
$\mathcal{V}_{\bs{r}}, \mathcal{N}, \mathcal{N}_{\bot}$ defined in
Lemma~\ref{lem:Affine-sub-space}.

\textit{1) $\Leftrightarrow$ 2): minimum norm vs. affine intersection.}
By Lemma~\ref{lem:Affine-sub-space}, the feasible set
$\mathcal{V}_{\bs{r}}$ is an affine subset of $L^2_w(I)$ with
direction $\mathcal{N}$, i.e., $\mathcal{V}_{\bs{r}} + \mathcal{N}
= \mathcal{V}_{\bs{r}}$. In particular, if $g_0 \in \mathcal{V}_{\bs{r}}$
is any fixed feasible point, then every $\mathring{g} \in \mathcal{V}_{\bs{r}}$
can be written uniquely as
\[
  \mathring{g} = g_0 + h, \quad h \in \mathcal{N}.
\]

Assume first that statement~2) holds, i.e.,
$\mathcal{V}_{\bs{r}} \cap \mathcal{N}_{\bot} = \{g_{\mathrm{plv}}\}$.
For any $g \in \mathcal{V}_{\bs{r}}$, write
$\mathring{g} = g_{\mathrm{plv}} + h$ with $h \in \mathcal{N}$. Since
$g_{\mathrm{plv}} \in \mathcal{N}_{\bot}$ and $h \in \mathcal{N}$,
we have $\langle g_{\mathrm{plv}}, h \rangle_w = 0$ and hence
\[
  \|\mathring{g}\|_w^2
  = \|g_{\mathrm{plv}} + h\|_w^2
  = \|g_{\mathrm{plv}}\|_w^2 + \|h\|_w^2
  \ge \|g_{\mathrm{plv}}\|_w^2,
\]
with equality only if $h = 0$, i.e., $\mathring{g} = g_{\mathrm{plv}}$.
Thus $g_{\mathrm{plv}}$ is the unique minimum-norm element in
$\mathcal{V}_{\bs{r}}$, which proves statement~1).

Conversely, suppose $g_{\mathrm{plv}}$ is the unique solution of the
minimum-norm problem~\eqref{eq:plv-min-norm}. Decompose
$g_{\mathrm{plv}}$ orthogonally with respect to $\mathcal{N}$ as
\[
  g_{\mathrm{plv}} = g_{\perp} + h,
  \quad g_{\perp} \in \mathcal{N}_{\bot},\; h \in \mathcal{N}.
\]
Since $\mathcal{V}_{\bs{r}}$ is an affine set with direction
$\mathcal{N}$, we have $g_{\perp} = g_{\mathrm{plv}} - h \in
\mathcal{V}_{\bs{r}}$. Moreover,
\[
  \|g_{\perp}\|_w^2
  = \|g_{\mathrm{plv}} - h\|_w^2
  = \|g_{\mathrm{plv}}\|_w^2 - \|h\|_w^2
  \le \|g_{\mathrm{plv}}\|_w^2,
\]
with strict inequality whenever $h \ne 0$. By the optimality of
$g_{\mathrm{plv}}$, it must hold that $h = 0$, and hence
$g_{\mathrm{plv}} = g_{\perp} \in \mathcal{N}_{\bot}$. Combined with
$g_{\mathrm{plv}} \in \mathcal{V}_{\bs{r}}$, this implies
$g_{\mathrm{plv}} \in \mathcal{V}_{\bs{r}} \cap \mathcal{N}_{\bot}$.
Finally, if there were two distinct elements
$g_1, g_2 \in \mathcal{V}_{\bs{r}} \cap \mathcal{N}_{\bot}$, then
$g_1 - g_2 \in \mathcal{N} \cap \mathcal{N}_{\bot} = \{0\}$, so
$g_1 = g_2$. Hence $\mathcal{V}_{\bs{r}} \cap \mathcal{N}_{\bot}$ is
a singleton, and statement~2) holds. This proves $1) \Leftrightarrow 2)$.

\textit{2) $\Leftrightarrow$ 3): geometric formulation vs. trigonometric parameterization.}
Lemma~\ref{lem:Affine-sub-space} further shows that
$\mathcal{N}_{\bot}$ is a $(2M-1)$-dimensional subspace spanned by
the trigonometric basis
\[
  \Big\{ 1,\,
         \cos(\kappa_m x),\,
         \sin(\kappa_m x)
         \;\big|\;
         m = 1,\ldots,M-1 \Big\},
\]
so every $g \in \mathcal{N}_{\bot}$ admits a unique representation of
the form~\eqref{eq:plv-trig-form} with some
$\bs{b} \in \mathbb{R}^{2M-1}$. In particular, any
$g_{\mathrm{plv}} \in \mathcal{V}_{\bs{r}} \cap \mathcal{N}_{\bot}$
must be representable as in~\eqref{eq:plv-trig-form}, and the
constraints defining $\mathcal{V}_{\bs{r}}$ translate exactly into the
feasibility conditions~\eqref{eq:plv-trig-feas}. This yields
statement~3).

Conversely, if $g_{\mathrm{plv}}$ admits the representation
\eqref{eq:plv-trig-form} with a vector $\bs{b}$ that
satisfies~\eqref{eq:plv-trig-feas}, then by construction
$g_{\mathrm{plv}} \in \mathcal{N}_{\bot}$ and
$\langle g_{\mathrm{plv}}, e^{\ii\kappa_m(\cdot)} \rangle_w = r_m$
for all $-M < m < M$, i.e., $g_{\mathrm{plv}} \in \mathcal{V}_{\bs{r}}$.
Uniqueness of $\bs{b}$ (and hence of $g_{\mathrm{plv}}$) follows from
the linear independence of the trigonometric basis. Therefore
$\mathcal{V}_{\bs{r}} \cap \mathcal{N}_{\bot}
= \{g_{\mathrm{plv}}\}$, and statement~2) holds. This proves
$2) \Leftrightarrow 3)$.

\textit{3) $\Leftrightarrow$ 4): trigonometric form vs. closed-form coefficients.}
Substituting the expansion~\eqref{eq:plv-trig-form} into the
covariance constraints
$\langle g_{\mathrm{plv}}, e^{\ii\kappa_m(\cdot)} \rangle_w = r_m$,
$-M < m < M$, and separating real and imaginary parts yield a linear
system of the form
\begin{equation}
  \mathbf{G}\,\bs{b} = \pi \bs{y},
  \label{eq:plv-linear-system}
\end{equation}
where $\bs{y} \in \mathbb{R}^{2M-1}$ is defined
in~\eqref{eq:linear-reg} and $\mathbf{G}$ is the block-diagonal
Gram matrix in~\eqref{eq:plv-G-def}. More
precisely, the entries of $\mathbf{G}_{\Re}$ and $\mathbf{G}_{\Im}$
can be written as
\[
\begin{split}
  [\mathbf{G}_{\Re}]_{m,n}
  &= \langle \cos(\kappa_m \cdot), \cos(\kappa_n \cdot) \rangle_w,
  \\
  [\mathbf{G}_{\Im}]_{m-1,n-1}
  &= \langle \sin(\kappa_m \cdot), \sin(\kappa_n \cdot) \rangle_w.
  \end{split}
\]
For any $\bs{c} \in \mathbb{R}^{M}$, we have
\[
  \bs{c}^{\top} \mathbf{G}_{\Re} \bs{c}
  = \big\| \sum_{m=0}^{M-1} c_m \cos(\kappa_m \cdot) \big\|_w^2
  \ge 0,
\]
with equality if and only if
$\sum_{m=0}^{M-1} c_m \cos(\kappa_m x) = 0$ for all $x \in I$,
which by linear independence of the cosine functions implies
$\mathbf{c} = \mathbf{0}$. Hence $\mathbf{G}_{\Re}$ is symmetric
positive definite. An analogous argument shows that
$\mathbf{G}_{\Im}$ is symmetric positive definite. Thus
$\mathbf{G}$ is positive definite and invertible, and
\eqref{eq:plv-linear-system} has the unique solution
$\bs{b} = \pi \mathbf{G}^{-1} \mathbf{y}$, which gives the
closed-form expression~\eqref{eq:plv-closed-form}. This proves
statement~4) given statement~3).

Conversely, if $\bs{b} = \pi \mathbf{G}^{-1} \mathbf{y}$, then
\eqref{eq:plv-linear-system} holds, and thus 
$g_{\mathrm{plv}}(\cdot;\bs{b})$ of the
form~\eqref{eq:plv-trig-form} satisfies covariance constraints
and belongs to $\mathcal{N}_{\bot}$. This recovers statement~3),
and hence $3) \Leftrightarrow 4)$.

Combining the implications $1) \Leftrightarrow 2)$,
$2) \Leftrightarrow 3)$, and $3) \Leftrightarrow 4)$ yields the
equivalence of all four characterizations.
\hfill$\blacksquare$

\section{Proof of Lemma~\ref{lem:sdp_cheb}}
\label{appendix-pf-lemma-SDP}
The argument relies on the following classical result (see also~\cite{RohVandenberghe} and~\cite[Sec. 3.2]{krejn1977markov}).
\begin{proposition}
	\label{prop:Lukas}
(Markov--Luk\'acs Theorem~\cite[Theorem. 1.21.1]{szego2003orthogonal}).
Let $g_{\le p}(x)$ be a degree-$p$ Chebyshev polynomial on
$[-1,1]$ with odd $p$, and suppose that $g_{\le p}(x) \ge 0$ for
all $x \in [-1,1]$. Then $g_{\le p}$ admits a representation of
the form
\begin{equation}
  g_{\le p}(x)
  = (1+x) f_{\le \lfloor p/2 \rfloor}(x)^2
    + (1-x) h_{\le \lfloor p/2 \rfloor}(x)^2,
  \label{eq:ML}
\end{equation}
for some polynomials
$f_{\le \lfloor p/2 \rfloor}, h_{\le \lfloor p/2 \rfloor}$ of degree at most
$\lfloor p/2 \rfloor$.
\end{proposition}
Our goal is to show that:
\(
g_{\le p}(\cdot;\bs a) \ge 0 \text{ on } I \iff \bs a \in \mathcal D.
\)

\emph{Orthonormal basis}.
Following Lemma~1, define the \emph{orthonormal} polynomials
$\{\psi_n\}_{n\ge 0}$ on $I$ with respect to $w$ by
\begin{equation}
  \psi_n(x) := \sqrt{\frac{2 - \delta_{n0}}{\pi}}\,T_n(x),\ n=0,1,\ldots,
  \label{eq:psi-def-app}
\end{equation}
so that $\{\psi_n\}$ is orthonormal in $L_w^2(I)$,
i.e., $\langle \psi_m,\psi_n\rangle_w = \delta_{mn}$. Then
$g_{\le p}(x)$ admits the expansion
\begin{equation}
  g_{\le p}(x)
  = \sum_{n=0}^p \zeta_n \psi_n(x),
  \label{eq:g-psi-exp}
\end{equation}
for some $\bs \zeta = (\zeta_0,\ldots,\zeta_p)^\top$.
Using \eqref{eq:psi-def-app} and the Chebyshev expansion
$g_{\le p}(x;a) = \sum_{n=0}^p a_n T_n(x)$, we have
\begin{equation}
  \bs \zeta = \sqrt{\tfrac{\pi}{2}}\, \bs a.
  \label{eq:zeta-a-relation}
\end{equation}
Thus it suffices to characterize nonnegativity in terms of
$\bs \zeta$.

{1) \emph{Sufficiency}: nonnegativity $\Rightarrow a \in \mathcal{D}$}.
Assume $g_{\le p}(x;a) \ge 0$ for all $x \in I$.
By Proposition~3, there exist polynomials
$f_{\le \lfloor p/2\rfloor}$, $h_{\le \lfloor p/2\rfloor}$ with degree up to
$\lfloor p/2 \rfloor$ such that
\begin{equation}
  g_{\le p}(x;a)
  = (1+x) f_{\le \lfloor p/2\rfloor}(x)^2
    + (1-x) h_{\le \lfloor p/2\rfloor}(x)^2.
  \label{eq:ML-2}
\end{equation}

Expand $f_{\le \lfloor p/2\rfloor}$,
$h_{\le \lfloor p/2\rfloor}$ in the orthonormal basis
$\{\psi_n\}$ as
\begin{align}
  f_{\le \lfloor p/2\rfloor}(x)
  &= \sum_{n=0}^{p_{\mathrm{half}}-1}  u_n \psi_n(x)
   = \bs u^\top \bs \psi(x),
   \label{eq:f-expansion} \\
  h_{\le \lfloor p/2\rfloor}(x)
  &= \sum_{n=0}^{p_{\mathrm{half}}-1} v_n \psi_n(x)
   = \bs v^\top \bs \psi(x),
   \label{eq:h-expansion}
\end{align}
where 
\(
  \bs \psi(x) :=
  [
    \psi_0(x) , \cdots , \psi_{p_{\mathrm{half}}-1}(x)
  ]^\top \in \mathbb R^{\half},
\)
and 
\[
  \bs u = [u_0,\ldots,u_{p_{\mathrm{half}}-1}]^\top,\quad
  \bs v = [v_0,\ldots,v_{p_{\mathrm{half}}-1}]^\top.
\]
Define the rank-one positive semidefinite matrices
\begin{equation}
  \bs S_1 := \bs u \bs u^\top,\quad \bs S_2 := \bs v \bs v^\top.
  \label{eq:S1S2-def}
\end{equation}
Then we obtain
\(
  f_{\le \lfloor p/2\rfloor}(x)^2
  = \bs \psi(x)^\top \bs S_1 \bs \psi(x),
  h_{\le \lfloor p/2\rfloor}(x)^2
  = \bs \psi(x)^\top \bs S_2 \bs \psi(x),
\)
and substituting these into \eqref{eq:ML-2} yields
\begin{equation}
	\begin{aligned}
		g_{\leq p}(x) 
		&= (1+x) \bm{\psi}(x)^\top \bs{S}_1 \bm{\psi}(x)
		\\
		&+ (1-x) \bm{\psi}(x)^\top \bs{S}_2 \bm{\psi}(x).
	\end{aligned}
	\label{eq:g-psi-S1S2}
\end{equation}

To relate $\bs S_1, \bs S_2$ to the coefficient vector $\bs a$, we follow
the construction in Definition~\ref{def:nonnegative-cone}. Let $\{\nu_i\}_{i=0}^N \subset I$
denote the Chebyshev zero nodes of $T_{N+1}$ (cf.
Definition~\ref{def:zeros-nodes}), and define the orthonormal DCT-II matrix
$\bs \Psi \in \mathbb{R}^{(N+1)\times (N+1)}$ by
\begin{equation}
  [\bs \Psi]_{j,n}
  := 
  \sqrt{c_N} \bs{\psi}_n(\nu_j),
  \quad 0 \le j,n \le N,
  \label{eq:DCT-def}
\end{equation}
where $c_N={\pi}/({N+1})$. \eqref{eq:DCT-def} is equivalent to the expression in Lemma~4 via \eqref{eq:psi-def-app}. As shown in~\cite[Lemma 3.47]{Plonka2023Chebyshev} (see also~\cite[Sec. 8.5]{briggs1995dft}),
$\bs \Psi$ is orthonormal, i.e., $\bs \Psi^{-1} = \bs \Psi^\top$.

For $N \ge p$, define the truncated matrices
\begin{subequations}
\label{eq:trunc-Psi}
\begin{align}
  \bs \Psi_{p_{\mathrm{half}}}
  &:= \bs \Psi(:,0{:}p_{\mathrm{half}}-1)
    \in \mathbb{R}^{(N+1)\times p_{\mathrm{half}}}, \\
  \bs \Psi_{p+1}
  &:= \bs \Psi(:,0{:}p)
    \in \mathbb{R}^{(N+1)\times (p+1)}.
\end{align}
\end{subequations}
Evaluating \eqref{eq:g-psi-S1S2} at 
$\{\nu_j\}_{j=0}^N$ and collecting them in vector form:
\[
  g_{\le p}(\bs \nu)
  :=
  [
    g_{\le p}(\nu_0) , \cdots , g_{\le p}(\nu_N)
  ]^\top.
\]
We obtain (cf.~\eqref{eq:g-psi-S1S2})
\begin{equation}
	\label{eq:g-nu-S1S2}
	\begin{aligned}
		c_N \cdot g_{\leq p}(\bs{\nu}) = &( \bs{1}_{} + \bm{\nu} )\circ \diag(\bs{\Psi}_{\half} \bs{S}_1 \bs{\Psi}_{\half}^\top ) \\
		+& ( \bs{1}_{} - \bm{\nu} ) \circ \diag(\bs{\Psi}_{\half} \bs{S}_2 \bs{\Psi}_{\half}^\top ),
	\end{aligned}
\end{equation}
where $\bs \nu := [\nu_0,\ldots,\nu_N]^\top$, “$\circ$” denotes the
Hadamard product, and $\operatorname{diag}(\cdot)$ extracts the
diagonal of a square matrix as a vector.

On other hand, combining the expansion
\eqref{eq:g-psi-exp} with the definition of $\bs \Psi$ in
\eqref{eq:DCT-def} yields the sampling relation
\begin{equation}
  \sqrt{c_N}\, \cdot 
  g_{\le p}(\bs \nu)
  = \bs \Psi_{p+1}\, \bs \zeta,
  \label{eq:g-nu-zeta}
\end{equation}
where $\bs \zeta$ is defined in
\eqref{eq:g-psi-exp}. Since $\bs \Psi$ is orthonormal, the columns
of $\bs \Psi_{p+1}$ are orthonormal as well, so that
$\bs \Psi_{p+1}^\top \bs \Psi_{p+1} = \bs I$. Hence, using~\eqref{eq:zeta-a-relation}, we may invert
\eqref{eq:g-nu-zeta} to obtain
\begin{equation}
  \bs a
  = \sqrt{2c_N/\pi}\, \cdot 
    \bs \Psi_{p+1}^\top g_{\le p}(\nu).
  \label{eq:zeta-from-g}
\end{equation}

Substituting \eqref{eq:g-nu-S1S2} into \eqref{eq:zeta-from-g}
and combining with \eqref{eq:zeta-a-relation} gives
\begin{equation}
\begin{aligned}
  \bs a
  &= \bs \Psi_{p+1}^\top
     \big[
       (\bs 1+\bs \nu)\circ
       \operatorname{diag}(\bs \Psi_{p_{\mathrm{half}}} \bs S_1 \bs \Psi_{p_{\mathrm{half}}}^\top)
       \\
  &\hphantom{=\Psi_{p+1}^\top\big[}
       +
       (\bs 1- \bs \nu)\circ
       \operatorname{diag}(\bs \Psi_{p_{\mathrm{half}}} \bs S_2 \bs \Psi_{p_{\mathrm{half}}}^\top)
     \big].
\end{aligned}
  \label{eq:a-beta-S1S2}
\end{equation}
Up to the constant rescaling of $\bs S_1$, $\bs S_2$ (absorbed into
their definitions), the right-hand side of
\eqref{eq:a-beta-S1S2} is exactly the linear map $\beta$ in
Definition~\ref{def:nonnegative-cone}.
Thus $\bs a \in \mathcal{D}$. We have shown that nonnegativity of
$g_{\le p}(\cdot; \bs a)$ implies $\bs a \in \mathcal{D}$.

2) \emph{Necessity:} $\bs a \in \mathcal{D}$ $\Rightarrow$ nonnegativity.
Conversely, suppose that $\bs a \in \mathcal{D}$, i.e., there exist
$\bs S_1,\bs S_2 \in \mathbb{S}_{p_{\mathrm{half}}}^+$ such that
$\bs a = \beta(\bs S_1, \bs S_2)$. By the explicit form of $\beta$ in
\eqref{eq:map-define}, this means that the polynomial
$g_{\le p}(x;\bs a)$ can be written in the form of~\eqref{eq:g-psi-S1S2}.

Factorize $\bs S_1$, $\bs S_2$ as
\(
  \bs S_1 = \sum_{k=1}^{s_1} \bs u_k \bs u_k^\top,
  \bs S_2 = \sum_{\ell=1}^{s_2} \bs v_\ell \bs v_\ell^\top,
\)
for some $\bs u_k, \bs v_\ell \in \mathbb{R}^{p_{\mathrm{half}}}$. Define
scalar functions
\[
  f_k(x) := \bs u_k^\top \bs \psi(x),\quad
  h_\ell(x) := \bs v_\ell^\top \bs \psi(x),
\]
for $k=1,\ldots,s_1$ and $\ell=1,\ldots,s_2$. Then
\[
  \bs \psi(x)^\top \bs S_1 \bs \psi(x)
  = \sum_{k=1}^{s_1} f_k(x)^2,\
  \bs \psi(x)^\top \bs S_2 \bs \psi(x)
  = \sum_{\ell=1}^{s_2} h_\ell(x)^2,
\]
and $g_{\le p}(x;\bs a)$ can be rewritten as: for $x \in I$,
\begin{equation}
  g_{\le p}(x; \bs a)
  =
  (1+x)\sum_{k=1}^{s_1} f_k(x)^2
  +
  (1-x)\sum_{\ell=1}^{s_2} h_\ell(x)^2.
  \label{eq:g-sos-final}
\end{equation}
For all $x \in I = [-1,1]$, we have $1\pm x \ge 0$ and each
term $f_k(x)^2$, $h_\ell(x)^2$ is nonnegative. Thus the
right-hand side of \eqref{eq:g-sos-final} is nonnegative, and we conclude that
$g_{\le p}(x;a) \ge 0$ on $I$.

\emph{Semidefinite representability}.
By \eqref{eq:D-def}, the set $\mathcal{D}$ is the image
of the cone
$\mathbb{S}_{p_{\mathrm{half}}}^+ \times \mathbb{S}_{p_{\mathrm{half}}}^+$
under the linear map $\beta$. Hence $D$ is convex and
semidefinite representable, and the constraint
$g_{\le p}(x; \bs a) \ge 0$ is equivalent to the
existence of $\bs S_1, \bs S_2 \in \mathbb{S}_{p_{\mathrm{half}}}^+$
satisfying $\bs a = \beta(\bs S_1, \bs S_2)$. This proves Lemma~\ref{lem:sdp_cheb}.
\hfill$\blacksquare$

\bibliographystyle{IEEEtran}
\bibliography{IEEEabrv,refs1}

\begin{thebibliography}{10}
\providecommand{\url}[1]{#1}
\csname url@samestyle\endcsname
\providecommand{\newblock}{\relax}
\providecommand{\bibinfo}[2]{#2}
\providecommand{\BIBentrySTDinterwordspacing}{\spaceskip=0pt\relax}
\providecommand{\BIBentryALTinterwordstretchfactor}{4}
\providecommand{\BIBentryALTinterwordspacing}{\spaceskip=\fontdimen2\font plus
\BIBentryALTinterwordstretchfactor\fontdimen3\font minus
  \fontdimen4\font\relax}
\providecommand{\BIBforeignlanguage}[2]{{%
\expandafter\ifx\csname l@#1\endcsname\relax
\typeout{** WARNING: IEEEtran.bst: No hyphenation pattern has been}%
\typeout{** loaded for the language `#1'. Using the pattern for}%
\typeout{** the default language instead.}%
\else
\language=\csname l@#1\endcsname
\fi
#2}}
\providecommand{\BIBdecl}{\relax}
\BIBdecl

\bibitem{Miretti-2021}
L.~Miretti, R.~L.~G. Cavalcante, and S.~Sta{\'n}czak, ``Channel covariance
  conversion and modelling using infinite dimensional {Hilbert} spaces,''
  \emph{IEEE Trans. Signal Process.}, vol.~69, pp. 3145--3159, 2021.

\bibitem{Miretti-2018}
------, ``{FDD} massive {MIMO} channel spatial covariance conversion using
  projection methods,'' in \emph{Proc. IEEE ICASSP}, 2018, pp. 3609--3613.

\bibitem{Haghighatshoar-2018}
S.~Haghighatshoar, M.~B. Khalilsarai, and G.~Caire, ``Multi-band covariance
  interpolation with applications in massive {MIMO},'' in \emph{Proc. IEEE
  ISIT}, 2018, pp. 386--390.

\bibitem{Barzegar-2019}
M.~B. Khalilsarai, S.~Haghighatshoar, X.~Yi, and G.~Caire, ``{FDD} massive
  {MIMO} via {UL/DL} channel covariance extrapolation and active channel
  sparsification,'' \emph{IEEE Trans. Wireless Commun.}, vol.~18, no.~1, pp.
  121--135, 2019.

\bibitem{Bameri-2023}
S.~Bameri, K.~A. Almahrog, R.~H. Gohary, A.~El-Keyi, and Y.~A.~E. Ahmed,
  ``Uplink to downlink channel covariance transformation in {FDD} systems,''
  \emph{IEEE Trans. Signal Process.}, vol.~71, pp. 3196--3212, 2023.

\bibitem{fan2018angle}
D.~Fan, F.~Gao, Y.~Liu, Y.~Deng, G.~Wang, Z.~Zhong, and A.~Nallanathan, ``Angle
  domain channel estimation in hybrid millimeter wave massive {MIMO} systems,''
  \emph{IEEE Trans. Wireless Commun.}, vol.~17, no.~12, pp. 8165--8179, 2018.

\bibitem{cui2022channel}
M.~Cui and L.~Dai, ``Channel estimation for extremely large-scale {MIMO}:
  Far-field or near-field?'' \emph{IEEE Trans. Commun.}, vol.~70, no.~4, pp.
  2663--2677, 2022.

\bibitem{mohammadzadeh2022covariance}
S.~Mohammadzadeh, V.~H. Nascimento, R.~C. de~Lamare, and O.~Kukrer,
  ``Covariance matrix reconstruction based on power spectral estimation and
  uncertainty region for robust adaptive beamforming,'' \emph{IEEE Trans.
  Aerosp. Electron. Syst.}, vol.~59, no.~4, pp. 3848--3858, 2022.

\bibitem{shakya2025angular}
D.~Shakya, M.~Ying, and T.~S. Rappaport, ``Angular spread statistics for 6.75
  {GHz} {FR1} {(C)} and 16.95 {GHz FR3} mid-band frequencies in an indoor
  hotspot environment,'' in \emph{Proc. IEEE WCNC}, 2025, pp. 1--6.

\bibitem{Cavalcante-2020}
R.~L.~G. Cavalcante and S.~Sta{\'n}czak, ``Channel covariance estimation in
  multiuser massive {MIMO} systems with an approach based on infinite
  dimensional {Hilbert} spaces,'' in \emph{Proc. IEEE ICASSP}, 2020, pp.
  5180--5184.

\bibitem{kelner2019evaluation}
J.~M. Kelner and C.~Zi{\'o}{\l}kowski, ``Evaluation of angle spread and power
  balance for design of radio links with directional antennas in multipath
  environment,'' \emph{Phys. Commun.}, vol.~32, pp. 242--251, 2019.

\bibitem{haghighatshoar2016massive}
S.~Haghighatshoar and G.~Caire, ``Massive {MIMO} channel subspace estimation
  from low-dimensional projections,'' \emph{IEEE Trans. Signal Process.},
  vol.~65, no.~2, pp. 303--318, 2016.

\bibitem{fischer2025systematic}
G.~K. Fischer, T.~Schaechtle, A.~Gabbrielli, J.~Bordoy, I.~H{\"a}ring,
  F.~H{\"o}flinger, and S.~J. Rupitsch, ``A systematic survey and comparative
  analysis of angular-based indoor localization and positioning technologies,''
  \emph{IEEE Commun. Surveys Tuts.}, 2025.

\bibitem{sayeed2002deconstructing}
A.~M. Sayeed, ``Deconstructing multiantenna fading channels,'' \emph{IEEE
  Trans. Signal Process.}, vol.~50, no.~10, pp. 2563--2579, 2002.

\bibitem{You-2015}
L.~You, X.~Gao, X.-G. Xia, N.~Ma, and Y.~Peng, ``Pilot reuse for massive {MIMO}
  transmission over spatially correlated {Rayleigh} fading channels,''
  \emph{IEEE Trans. Wireless Commun.}, vol.~14, no.~6, pp. 3352--3366, 2015.

\bibitem{Cavalcante-Renato-2020}
R.~L.~G. Cavalcante and S.~Sta{\'n}czak, ``Hybrid data and model driven
  algorithms for angular power spectrum estimation,'' in \emph{Proc. IEEE
  GLOBECOM}, 2020, pp. 1--6.

\bibitem{song2020deep}
Y.~Song, M.~B. Khalilsarai, S.~Haghighatshoar, and G.~Caire, ``Deep learning
  for geometrically-consistent angular power spread function estimation in
  massive {MIMO},'' in \emph{Proc. IEEE GLOBECOM}, 2020, pp. 1--6.

\bibitem{Bauschke2006Extrapolation}
H.~H. Bauschke, P.~L. Combettes, and S.~G. Kruk, ``Extrapolation algorithm for
  affine-convex feasibility problems,'' \emph{Numer. Algorithms}, vol.~41, pp.
  239--274, 2006.

\bibitem{Agrawal-2021}
N.~Agrawal, R.~L.~G. Cavalcante, and S.~Sta{\'n}czak, ``Adaptive estimation of
  angular power spectra for time-varying {MIMO} channels,'' in \emph{Proc. IEEE
  SPAWC}, 2021, pp. 96--100.

\bibitem{Ninomiya-2022}
E.~Ninomiya, M.~Yukawa, R.~L.~G. Cavalcante, and L.~Miretti, ``Estimation of
  angular power spectrum using multikernel adaptive filtering,'' in \emph{Proc.
  APSIPA ASC}, 2022, pp. 1996--2000.

\bibitem{Kaneko-2024}
N.~Kaneko, M.~Yukawa, R.~L.~G. Cavalcante, and L.~Miretti, ``Robust estimation
  of angular power spectrum in massive {MIMO} under covariance estimation
  errors: Learning centers and scales of {Gaussians},'' in \emph{Proc. IEEE
  SPAWC}, 2024, pp. 576--580.

\bibitem{Interpretable-FLR-2009}
G.~M. James, J.~Wang, and J.~Zhu, ``Functional linear regression that's
  interpretable,'' \emph{Ann. Statist.}, vol.~37, no.~5A, pp. 2083--2108, 2009.

\bibitem{morris2015functional}
J.~S. Morris, ``Functional regression,'' \emph{Annu. Rev. Stat. Appl.}, vol.~2,
  no.~1, pp. 321--359, 2015.

\bibitem{TrefethenATAP}
L.~N. Trefethen, \emph{Approximation Theory and Approximation Practice}.\hskip
  1em plus 0.5em minus 0.4em\relax SIAM, 2013.

\bibitem{Plonka2023Chebyshev}
G.~Plonka, D.~Potts, G.~Steidl, and M.~Tasche, \emph{Numerical Fourier
  Analysis}.\hskip 1em plus 0.5em minus 0.4em\relax Springer, 2023.

\bibitem{szego2003orthogonal}
G.~Szeg{\"o}, \emph{Orthogonal Polynomials}, 4th~ed.\hskip 1em plus 0.5em minus
  0.4em\relax American Math. Soc., 1975.

\bibitem{Elad2007}
M.~Elad, P.~Milanfar, and R.~Rubinstein, ``Analysis versus synthesis in signal
  priors,'' \emph{Inverse Problems}, vol.~23, no.~3, pp. 947--968, 2007.

\bibitem{Stefan2010}
W.~Stefan, R.~Renaut, and A.~Gelb, ``Improved total variation-type
  regularization using higher-order edge detectors,'' \emph{SIAM J. Sci.
  Comput.}, vol.~32, no.~3, pp. 1287--1309, 2010.

\bibitem{bredies2010total}
K.~Bredies, K.~Kunisch, and T.~Pock, ``Total generalized variation,''
  \emph{SIAM J. Imaging Sci.}, vol.~3, no.~3, pp. 492--526, 2010.

\bibitem{chan2000high}
T.~Chan, A.~Marquina, and P.~Mulet, ``High-order total variation-based image
  restoration,'' \emph{SIAM J. Sci. Comput.}, vol.~22, no.~2, pp. 503--516,
  2000.

\bibitem{kim2009ell_1}
S.-J. Kim, K.~Koh, S.~Boyd, and D.~Gorinevsky, ``$\ell_1$ trend filtering,''
  \emph{SIAM Rev.}, vol.~51, no.~2, pp. 339--360, 2009.

\bibitem{Apostol}
T.~M. Apostol, \emph{Mathematical Analysis}, 2nd~ed.\hskip 1em plus 0.5em minus
  0.4em\relax Reading, MA: Addison-Wesley, 1974.

\bibitem{WatsonBessel}
G.~N. Watson, \emph{A Treatise on the Theory of {Bessel} Functions},
  2nd~ed.\hskip 1em plus 0.5em minus 0.4em\relax Cambridge Univ. Press, 1944.

\bibitem{briggs1995dft}
W.~L. Briggs and V.~E. Henson, \emph{The {DFT}: An Owner's Manual for the
  Discrete Fourier Transform}.\hskip 1em plus 0.5em minus 0.4em\relax SIAM,
  1995.

\bibitem{zhong2020fdd}
Z.~Zhong, L.~Fan, and S.~Ge, ``{FDD} massive {MIMO} uplink and downlink channel
  reciprocity properties: Full or partial reciprocity?'' in \emph{Proc. IEEE
  GLOBECOM}.\hskip 1em plus 0.5em minus 0.4em\relax IEEE, 2020, pp. 1--5.

\bibitem{HornJohnson2013}
R.~A. Horn and C.~R. Johnson, \emph{Matrix Analysis}, 2nd~ed.\hskip 1em plus
  0.5em minus 0.4em\relax Cambridge Univ. Press, 2013.

\bibitem{GolubVanLoan2013}
G.~H. Golub and C.~F.~V. Loan, \emph{Matrix Computations}, 4th~ed.\hskip 1em
  plus 0.5em minus 0.4em\relax Baltimore, MD, USA: Johns Hopkins Univ. Press,
  2013.

\bibitem{3GPP2024TR38901}
{3GPP}, ``Study on channel model for frequencies from 0.5 to 100 {GHz} (release
  18),'' 3GPP, Tech. Rep. TR 38.901 v18.0.0, 2024.

\bibitem{zhang2019theoretically}
H.~Zhang, Y.~Yu, J.~Jiao, E.~Xing, L.~El~Ghaoui, and M.~Jordan, ``Theoretically
  principled trade-off between robustness and accuracy,'' in \emph{Proc. Int.
  Conf. Mach. Learn. (ICML)}.\hskip 1em plus 0.5em minus 0.4em\relax PMLR,
  2019, pp. 7472--7482.

\bibitem{Qian-WeiYu-2025}
J.~Qian, S.~Salehkalaibar, J.~Chen, A.~Khisti, W.~Yu, W.~Shi, Y.~Ge, and
  W.~Tong, ``Rate-distortion-perception tradeoff for {Gaussian} vector
  sources,'' \emph{IEEE J. Sel. Areas Inf. Theory}, vol.~6, pp. 1--17, 2025.

\bibitem{akima1970new}
H.~Akima, ``A new method of interpolation and smooth curve fitting based on
  local procedures,'' \emph{J. ACM}, vol.~17, no.~4, pp. 589--602, 1970.

\bibitem{chi2011sensitivity}
Y.~Chi, L.~L. Scharf, A.~Pezeshki, and A.~R. Calderbank, ``Sensitivity to basis
  mismatch in compressed sensing,'' \emph{IEEE Trans. Signal Process.},
  vol.~59, no.~5, pp. 2182--2195, May 2011.

\bibitem{wang2022policy}
Y.~Wang and S.~Zou, ``Policy gradient method for robust reinforcement
  learning,'' in \emph{Proc. Int. Conf. Mach. Learn. (ICML)}.\hskip 1em plus
  0.5em minus 0.4em\relax PMLR, 2022, pp. 23\,484--23\,526.

\bibitem{liu2015robust}
H.~Liu, Y.~Li, and T.-K. Truong, ``Robust sparse signal reconstructions against
  basis mismatch and their applications,'' \emph{Inf. Sci.}, vol. 316, pp.
  1--17, Sep. 2015.

\bibitem{jaeckel2014quadriga}
S.~Jaeckel, L.~Raschkowski, K.~B{\"o}rner, and L.~Thiele, ``Quadriga: A {3-D}
  multi-cell channel model with time evolution for enabling virtual field
  trials,'' \emph{IEEE Trans. Antennas Propag.}, vol.~62, no.~6, pp.
  3242--3256, 2014.

\bibitem{Jaeckel2023QuaDRiGa}
S.~Jaeckel, L.~Raschkowski, K.~B{\"o}rner, L.~Thiele, F.~Burkardt, and
  E.~Eberlein, ``Quadriga: Quasi deterministic radio channel generator, user
  manual and documentation,'' Fraunhofer Heinrich Hertz Inst., Tech. Rep.
  v2.8.1, 2023.

\bibitem{NISTDLMF}
F.~W.~J. Olver, A.~B. Olde~Daalhuis, and et~al., ``{NIST} digital library of
  mathematical functions,'' \url{https://dlmf.nist.gov/}, 2023.

\bibitem{RohVandenberghe}
T.~Roh and L.~Vandenberghe, ``Discrete transforms, semidefinite programming,
  and sum-of-squares representations of nonnegative polynomials,'' \emph{SIAM
  J. Optim.}, vol.~16, no.~4, pp. 939--964, 2006.

\bibitem{krejn1977markov}
M.~G. Krejn, M.~G. Krein, A.~A. Nudel'man, and P.~L. {\v{C}}eby{\v{s}}ev,
  \emph{The Markov Moment Problem and Extremal Problems}, ser. Translations
  Math. Monographs, vol. 50.\hskip 1em plus 0.5em minus 0.4em\relax American
  Math. Soc., 1977.

\end{thebibliography}

\end{document}